\documentclass{llncs}
\pdfoutput=1

\usepackage{graphicx}
\usepackage{float}
\usepackage{rotating}

\DeclareGraphicsExtensions{.pdf,.png,.jpg}

\begin{document}

\title{Reversible Nets of Polyhedra}

\author{Jin Akiyama\inst{1} \and Stefan
  Langerman\inst{2}\thanks{Directeur de Recherches du F.R.S.-FNRS} \and Kiyoko Matsunaga\inst{1}}

\institute{Tokyo University of Science,\\
1-3 Kagurazaka, Shinjuku, Tokyo 162-8601, Japan\\
\email{ja@jin-akiyama.com}, \email{matsunaga@mathlab-jp.com}
\and
Universit\'e Libre de Bruxelles\\
Brussels, Belgium\\
\email{stefan.langerman@ulb.ac.be}
}

\maketitle

\begin{abstract}

An example of reversible (or hinge inside-out transformable) figures is the Dudeney's Haberdasher's puzzle in which an equilateral triangle is dissected into four pieces, then hinged like a chain, and then is transformed into a square by rotating the hinged pieces. Furthermore, the entire boundary of each figure goes into the inside of the other figure and becomes the dissection lines of the other figure. Many intriguing results on reversibilities of figures have been found in prior research, but most of them are results on polygons. This paper generalizes those results to a wider range of general connected figures. It is shown that two nets obtained by cutting the surface of an arbitrary convex polyhedron along non-intersecting dissection trees are reversible. Moreover, a condition for two nets of an isotetrahedron to be both reversible and tessellative is given.
\end{abstract}

\section{Introduction}\label{sec:Introduction}
A pair of hinged figures $P$ and $Q$ (see Fig.~\ref{fig:dudeney}) is said
to be \emph{reversible} (or \emph{hinge inside-out transformable}) if
$P$ and $Q$ satisfy the following conditions: 

\begin{enumerate}
\item There exists a dissection of $P$ into a finite number of pieces,
  $P_1, P_2, P_3, \ldots, P_n$.  
  A set of dissection lines or curves forms a tree. Such a tree is
  called a \emph{dissection tree}. 
\item Pieces $P_1, P_2, P_3, \ldots, P_n$ can be joined by $n-1$
  hinges located on the perimeter of $P$ like a chain. 
\item If one of the end-pieces of the chain is fixed and rotated, then
  the remaining pieces form $Q$ when rotated clockwise and $P$ when
  rotated counterclockwise. 
\item The entire boundary of $P$ goes into the inside of $Q$ and the
  entire boundary of $Q$ is composed exactly of the edges of the dissection
  tree of $P$. 
\end{enumerate}

\begin{figure}[h]
\parbox[t]{0.32\textwidth}{\vspace{1cm}\centering $P$: an equilateral triangle \\\includegraphics[width=\hsize]{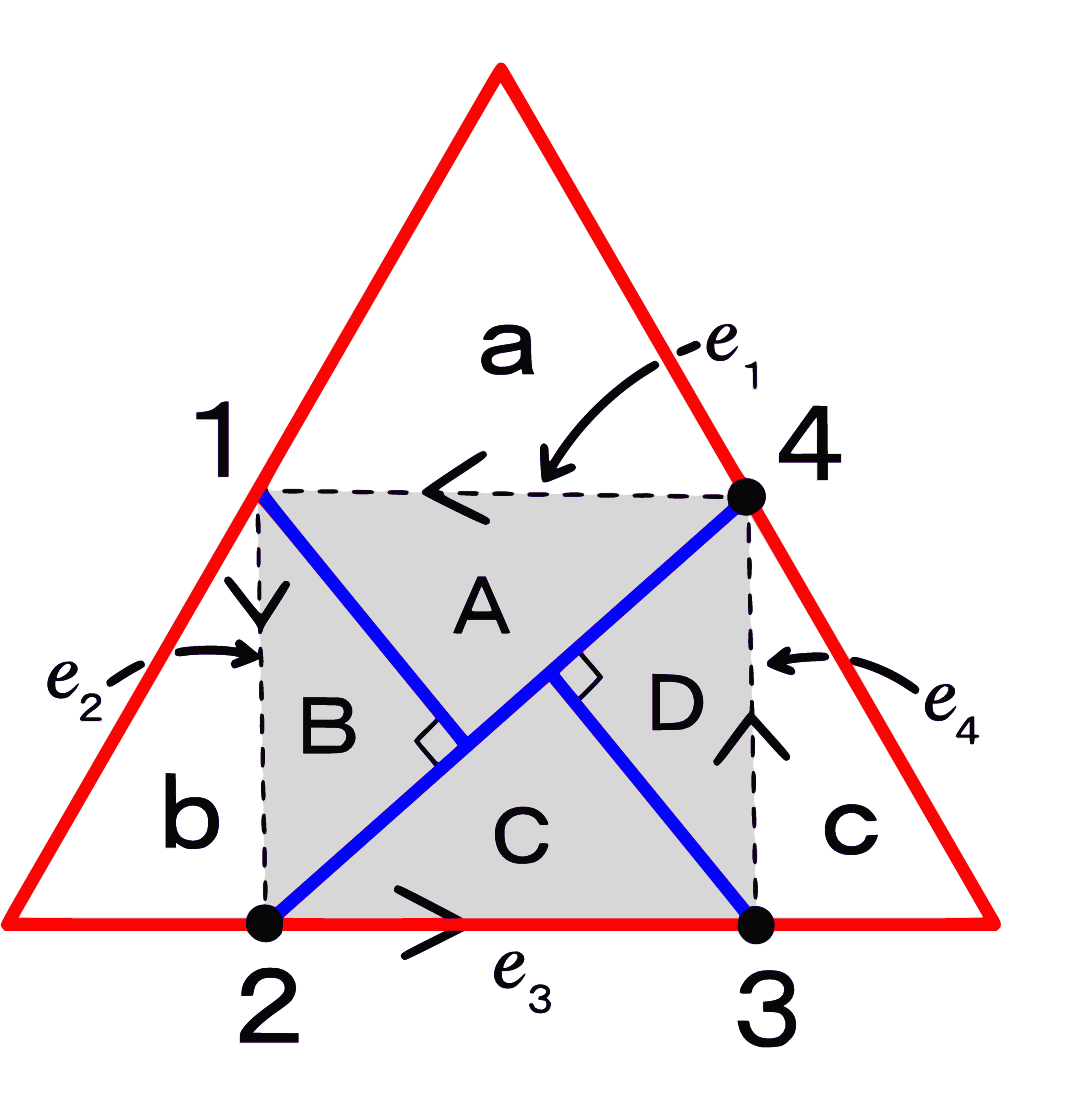}}
\parbox[t]{0.26\textwidth}{\centering ~\\\includegraphics[width=\hsize]{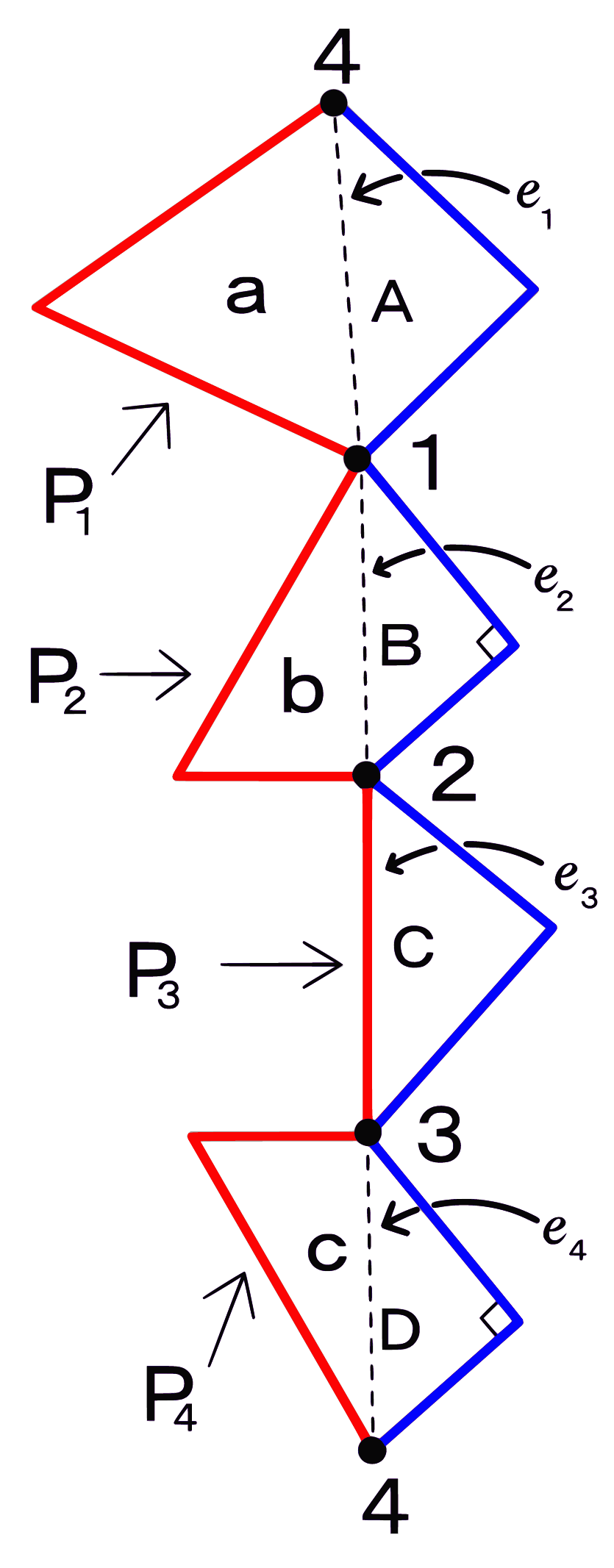}}
\parbox[t]{0.32\textwidth}{\vspace{1cm}\centering $Q$: a square \\\includegraphics[width=\hsize]{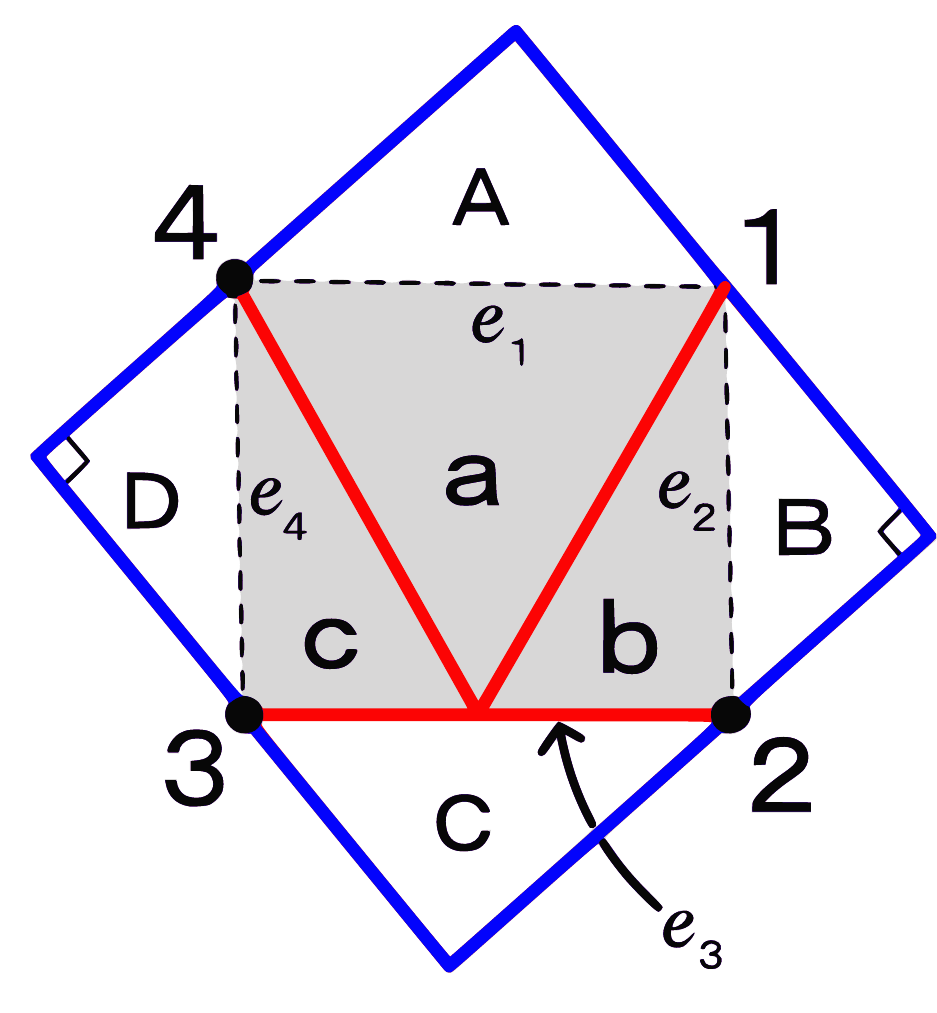}}
\caption{Reversible transformation between $P$ and $Q$.}\label{fig:dudeney}
\end{figure}

The theory of hinged dissections and reversibilities of figures has a
long history and the book by G.~N.~Frederickson~\cite{r11}
contains many interesting results. On the other hand, T.~Abbott
et. al~\cite{r1} proved 
that every pair of polygons $P$ and $Q$ with the same area is hinge
transformable if we don't require the reversible condition. When imposing the
reversible condition, hinge transformable figures have some remarkable
properties which were studied in~\cite{r3,r4,r5,r6,r9,Itoh2014}.

Let $T$ be a closed plane region whose perimeter consists of $n$
curved (or straight line) \emph{segments} $e_1, e_2, \ldots, e_n$  and let
these lines be labeled in clockwise order. Let $T'$ be a closed
region surrounded by the same segments $e_1, e_2, \ldots, e_n$ but in counterclockwise order. We then say that $T'$ is a \emph{conjugate region} of $T$ (Fig.~\ref{fig:conjugate}).

\begin{figure}[H]
\parbox[t]{0.3\textwidth}{\centering $T$ \\\includegraphics[width=0.8\hsize]{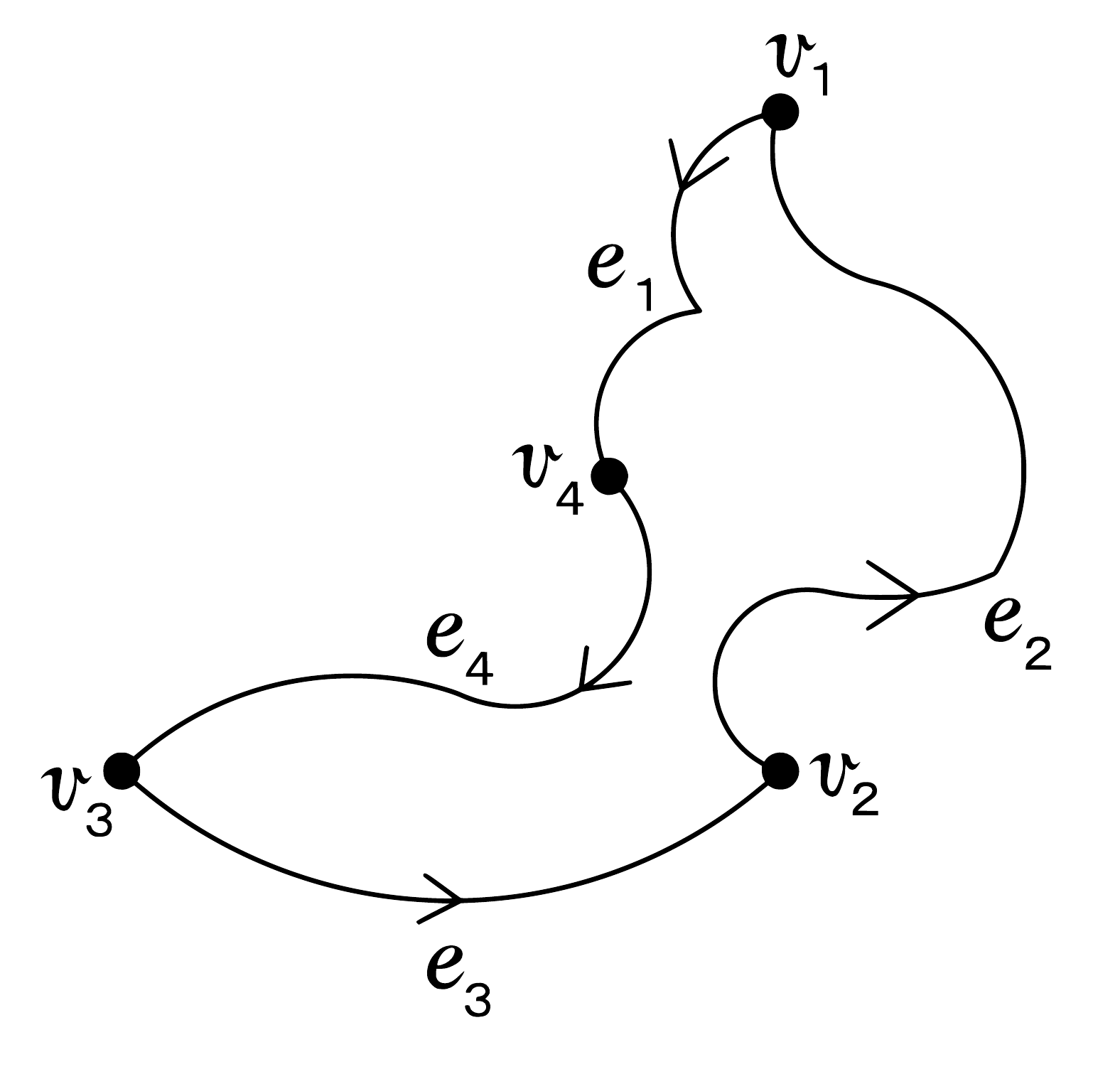}}
\parbox[t]{0.3\textwidth}{\centering ~\\\includegraphics[width=0.5\hsize]{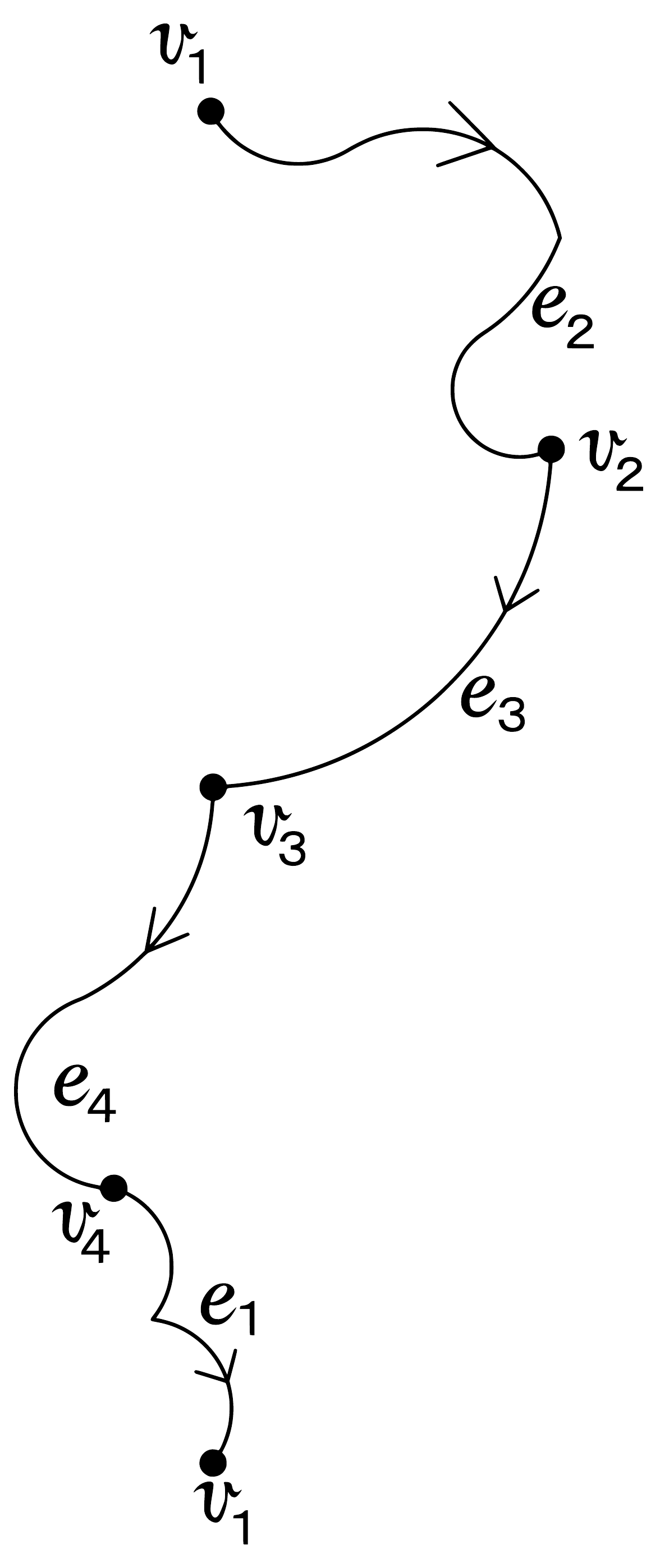}}
\parbox[t]{0.3\textwidth}{\centering $T'$ \\\includegraphics[width=0.8\hsize]{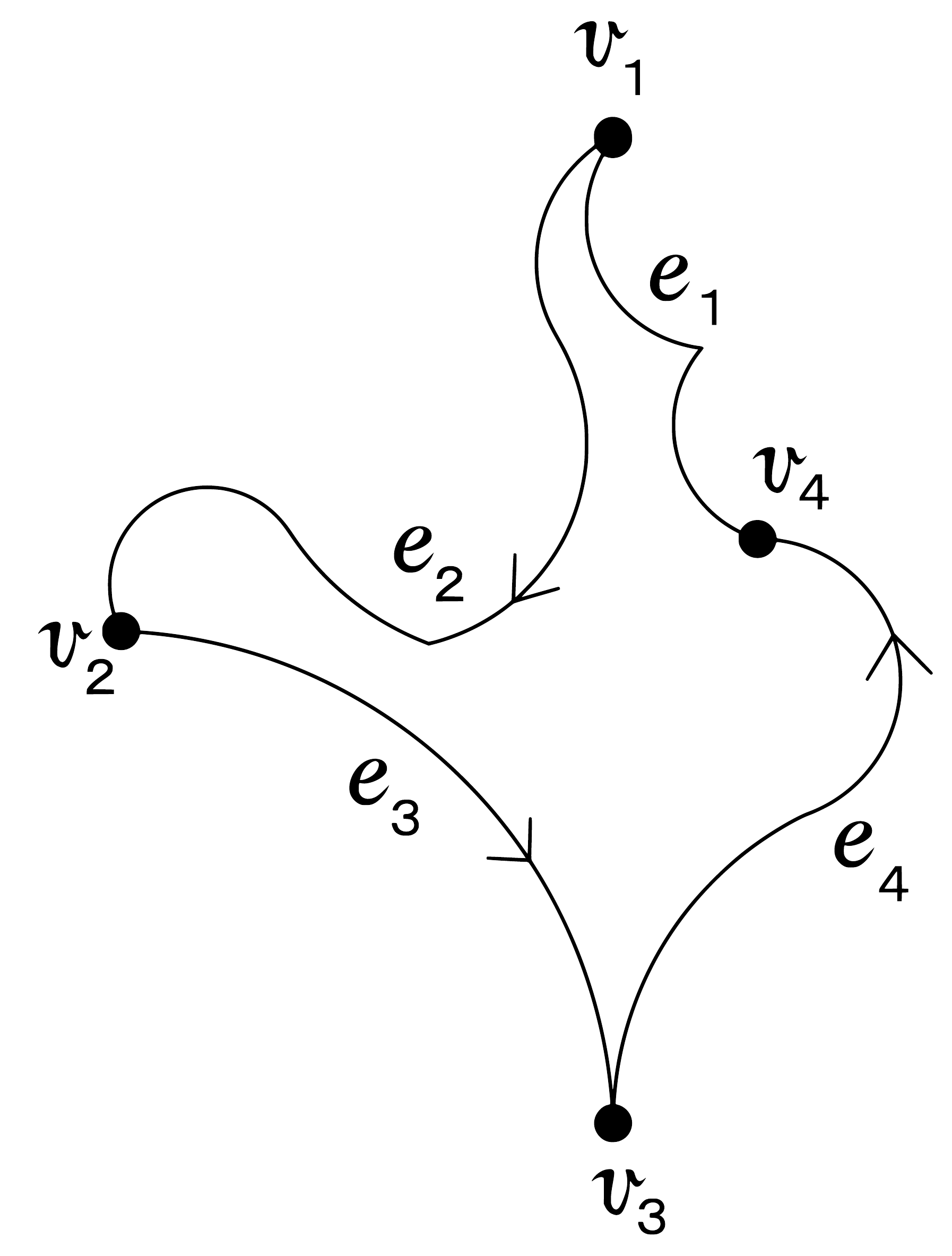}}
\caption{T and one of its conjugate regions.}\label{fig:conjugate}
\end{figure}

Let $P$ be a plane figure. A region $T$ with $n$ vertices
$v_1,\ldots,v_n$  and with $n$ perimeter parts $e_1,\ldots,e_n$  is
called an \emph{inscribed region} of $P$ if all vertices $v_i$
($i=1,\ldots,n$) are located on the 
perimeter of $P$ and $T\subseteq P$. 

A \emph{trunk} of $P$ is a special kind of inscribed region $T$ of
$P$. First, cut out an inscribed region $T$ from $P$
(Fig.~\ref{fig:trunk}(a)). Let $e_i$ ($i=1,\ldots,n$) be the perimeter
part of $T$ joining two vertices $v_{i-1}$ and $v_i$ of $T$, where
$v_0=v_n$. Denote by $P_i$ the piece located outside of $T$ that
contains the perimeter part $e_i$. 
Some $P_i$ may be empty (or just a part $e_i$). Then, hinge each pair of pieces
$P_i$ and $P_{i+1}$ at their common vertex $v_i$ ($1\leq i\leq n-1$);
this results in a chain of pieces $P_i$ ($i=1,2,\ldots,n$) of $P$
(Fig.~\ref{fig:trunk}(b)). 
The chain and $T$ are called $(T, T')$-chain of $P$, and \emph{trunk of
  $P$}, respectively, if an appropriate rotation of the chain forms
$T'$ which is one of the conjugate regions of $T$ with all pieces
$P_i$ packed inside $T'$ without overlaps or gaps. The chain $T'$ is called a
\emph{conjugate trunk} of $P$ (Fig.~\ref{fig:trunk}(c)). 

\begin{figure}[h]
\parbox[t]{0.31\textwidth}{\centering (a)\\
     \includegraphics[width=\hsize]{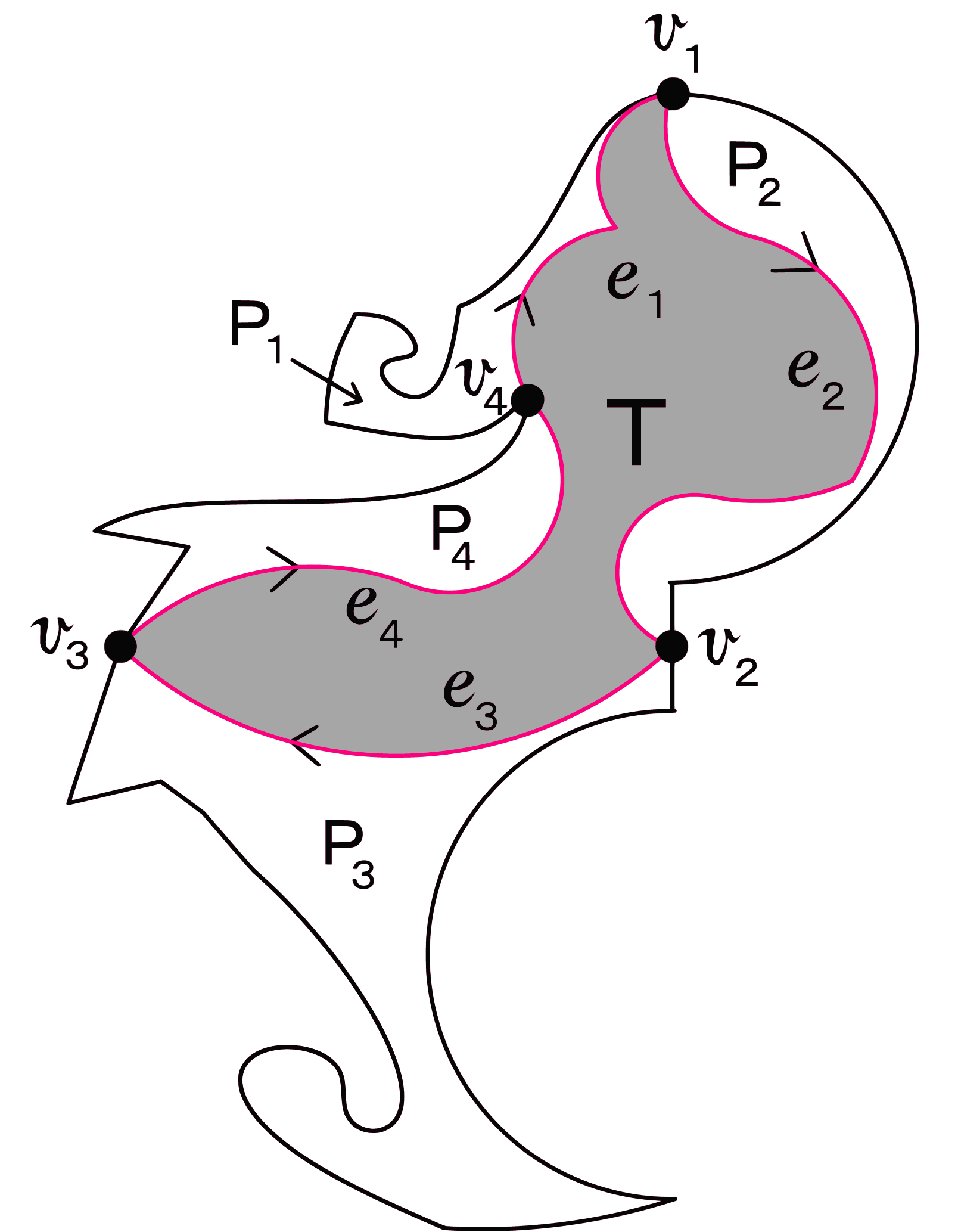} \\
Cut $T$ off from $P$\\$T$: the gray part.}
\parbox[t]{0.31\textwidth}{\centering (b)\\
    \includegraphics[width=1.1\hsize]{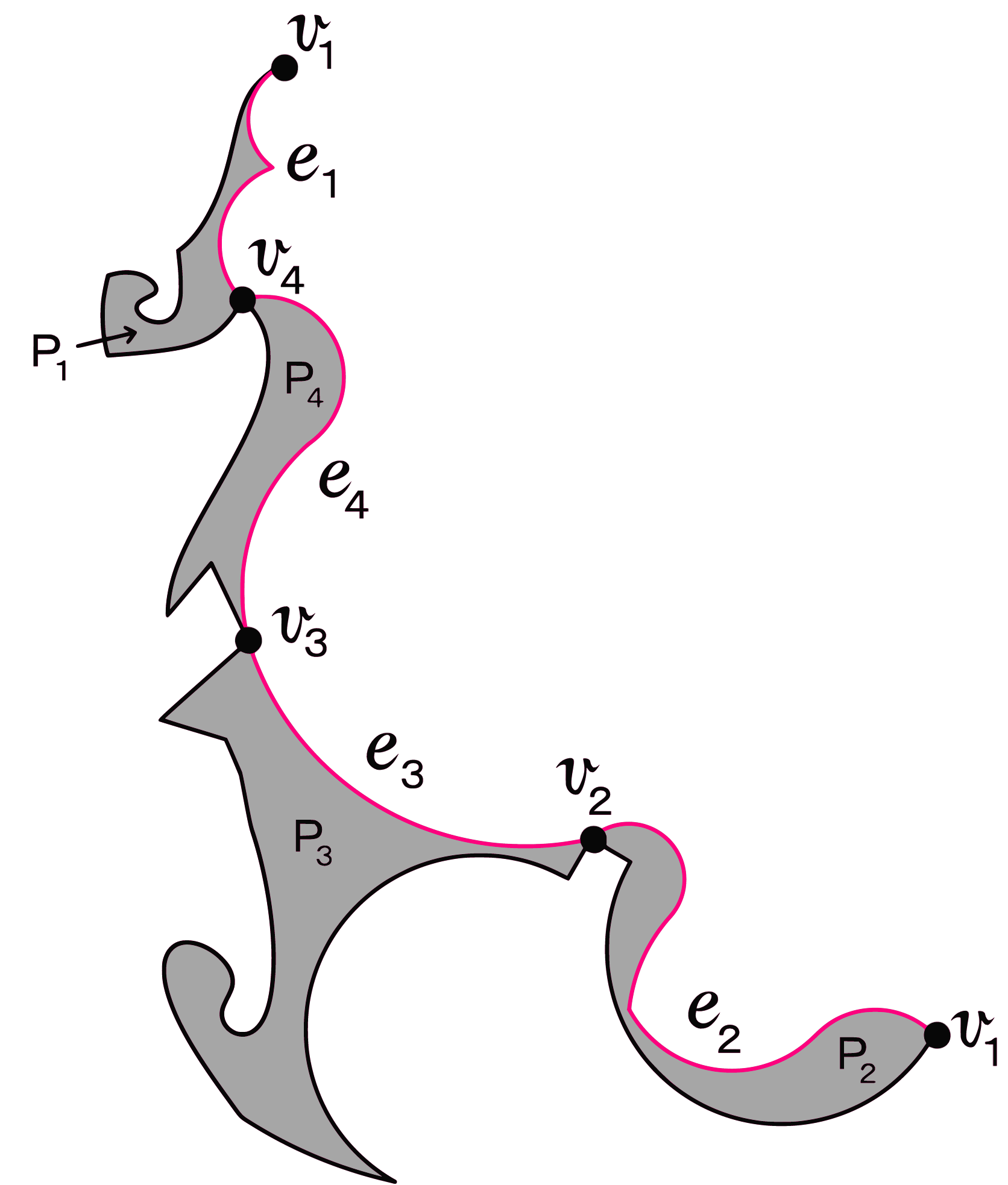}\\
    A $(T,T')$-chain of $P$}
\parbox[t]{0.31\textwidth}{\centering (c)\\
    \includegraphics[width=\hsize]{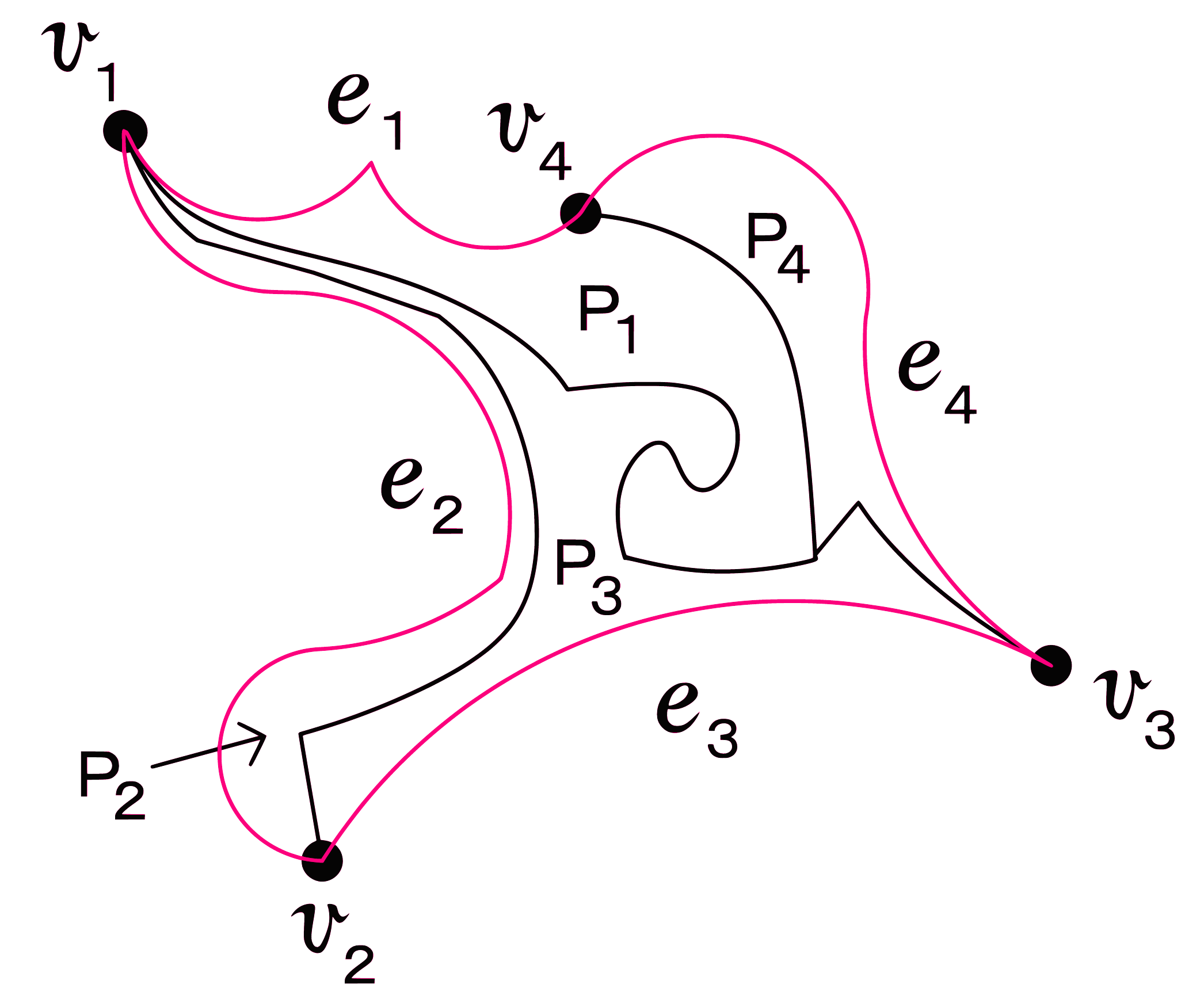}
    $T'$: a conjugate trunk of $P$}
\caption{A trunk $T$ of $P$, a $(T, T')$-chain of $P$ and a conjugate trunk $T'$ of $P$.}\label{fig:trunk}
\end{figure}

Suppose that a figure $P$ has a trunk $T$ and a conjugate trunk $T'$; and a
figure $Q$ has a trunk $T'$ and a conjugate trunk $T$. We then have two
chains, a $(T, T')$-chain of $P$ and a $(T', T)$-chain of $Q$ (Fig.~\ref{fig:chains}).

\begin{figure}[h]
\parbox[t]{0.30\textwidth}{\centering ~\\
     \includegraphics[width=\hsize]{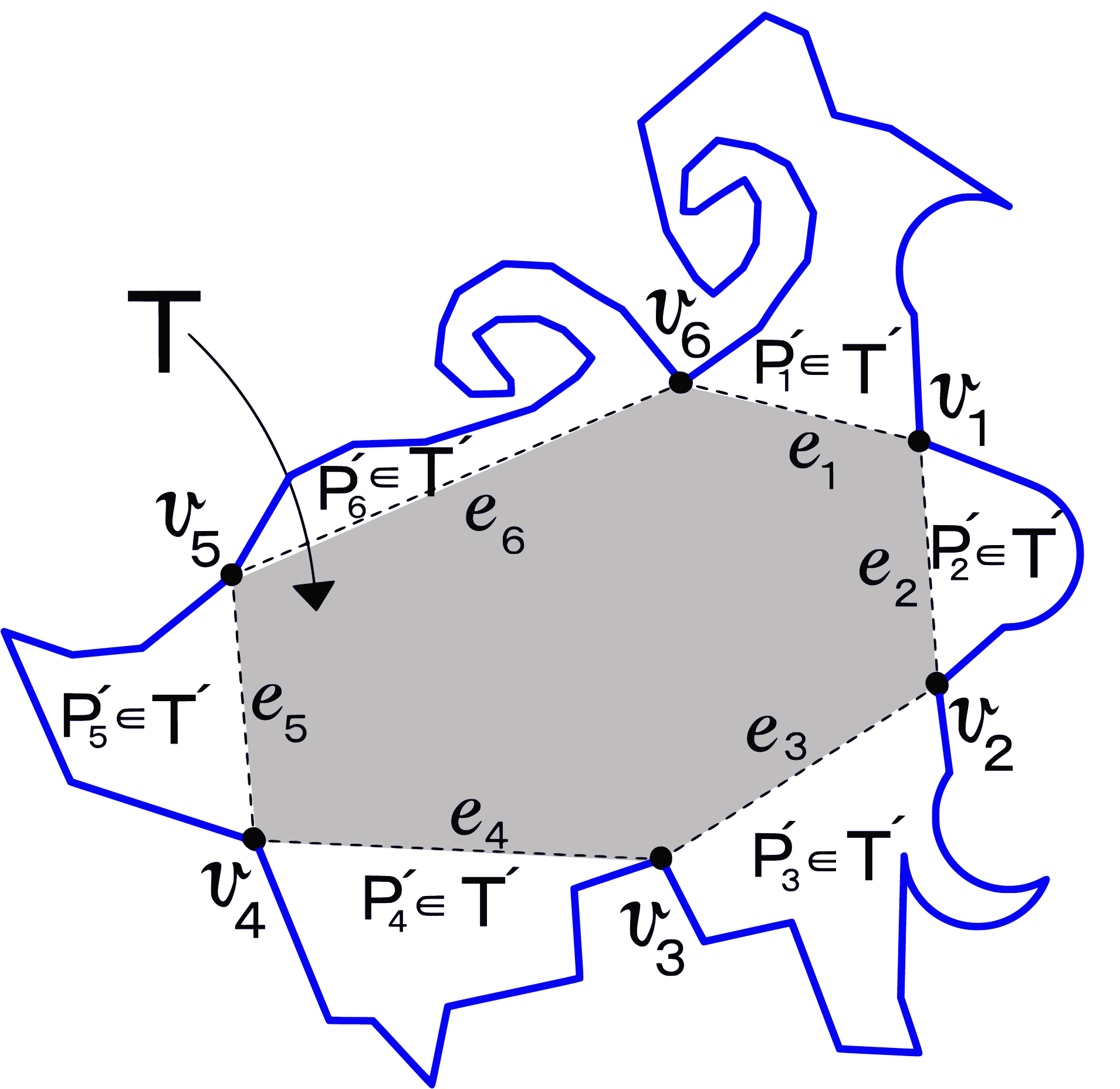} \\
    \vspace{5mm}
    {A $(T,T')$-chain of $P$}\\
    \vspace{5mm}
    {$T'$:}
     \includegraphics[width=\hsize]{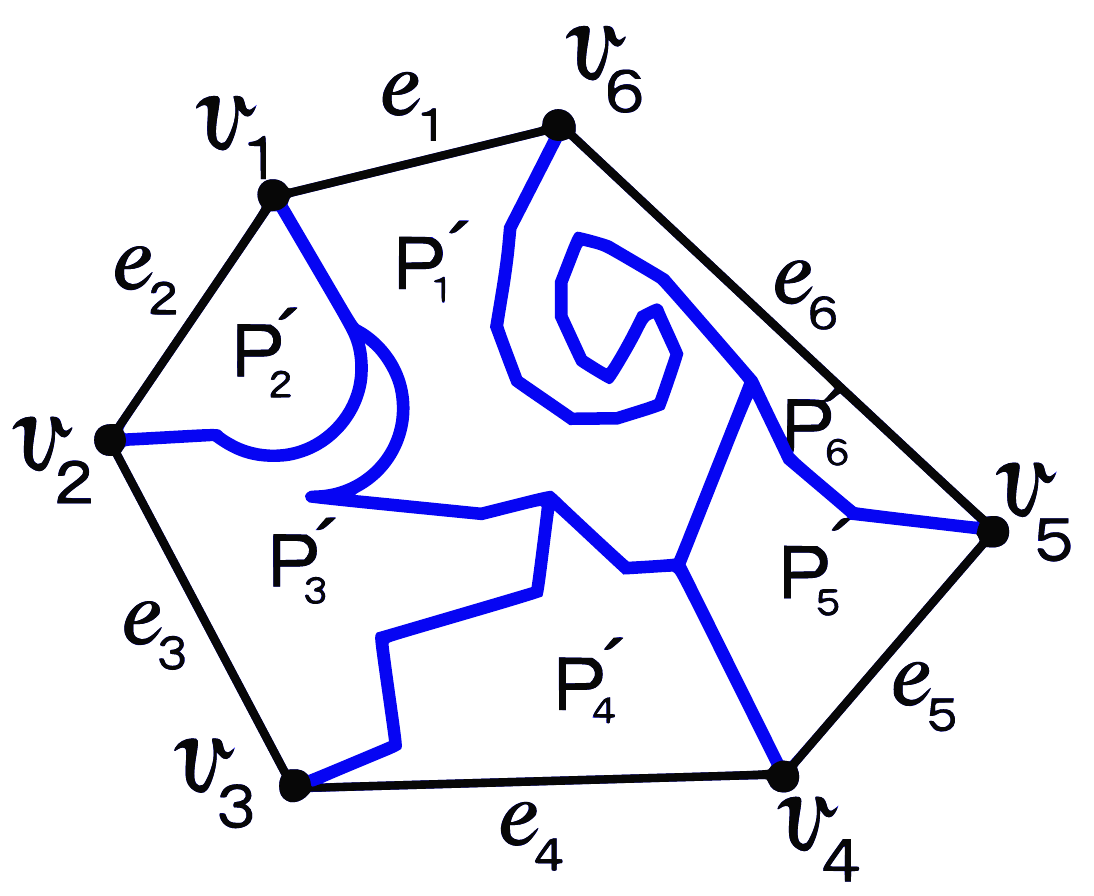} }
\parbox[t]{0.33\textwidth}{\centering ~\\
    \includegraphics[width=0.4\hsize]{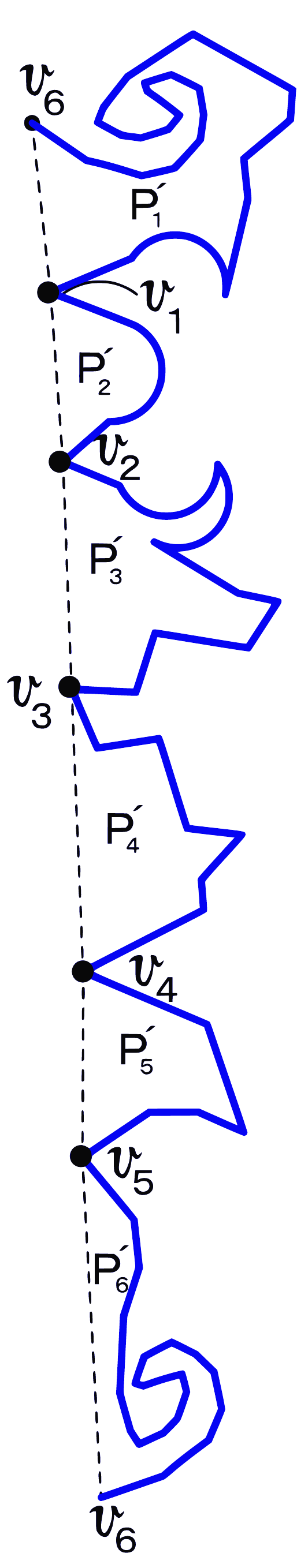}
    \includegraphics[width=0.55\hsize]{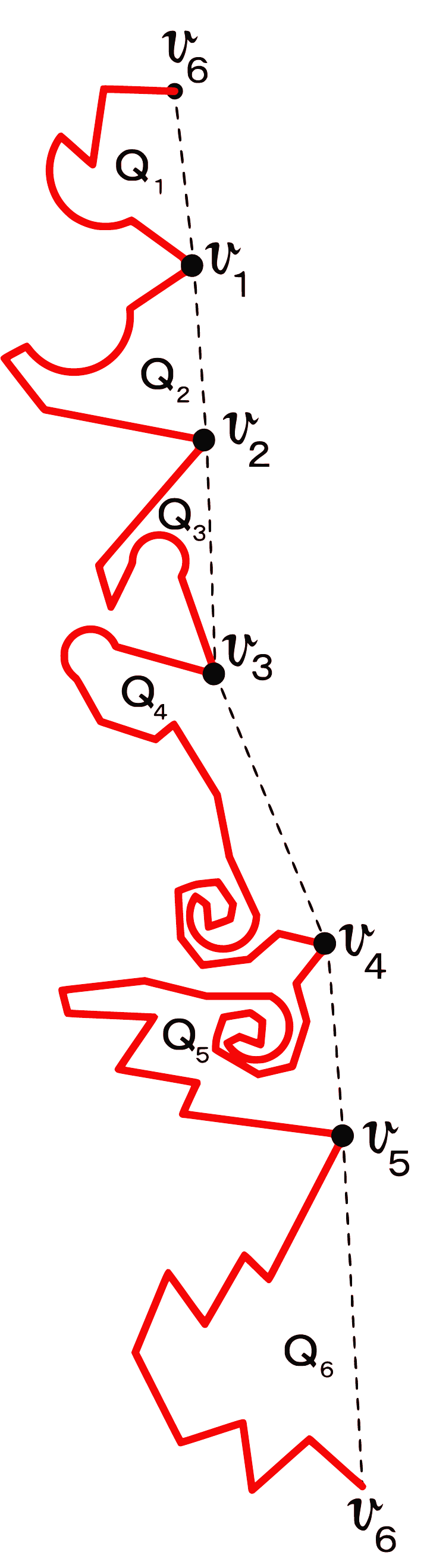}}
\parbox[t]{0.30\textwidth}{\centering ~\\
     \includegraphics[width=1.1\hsize]{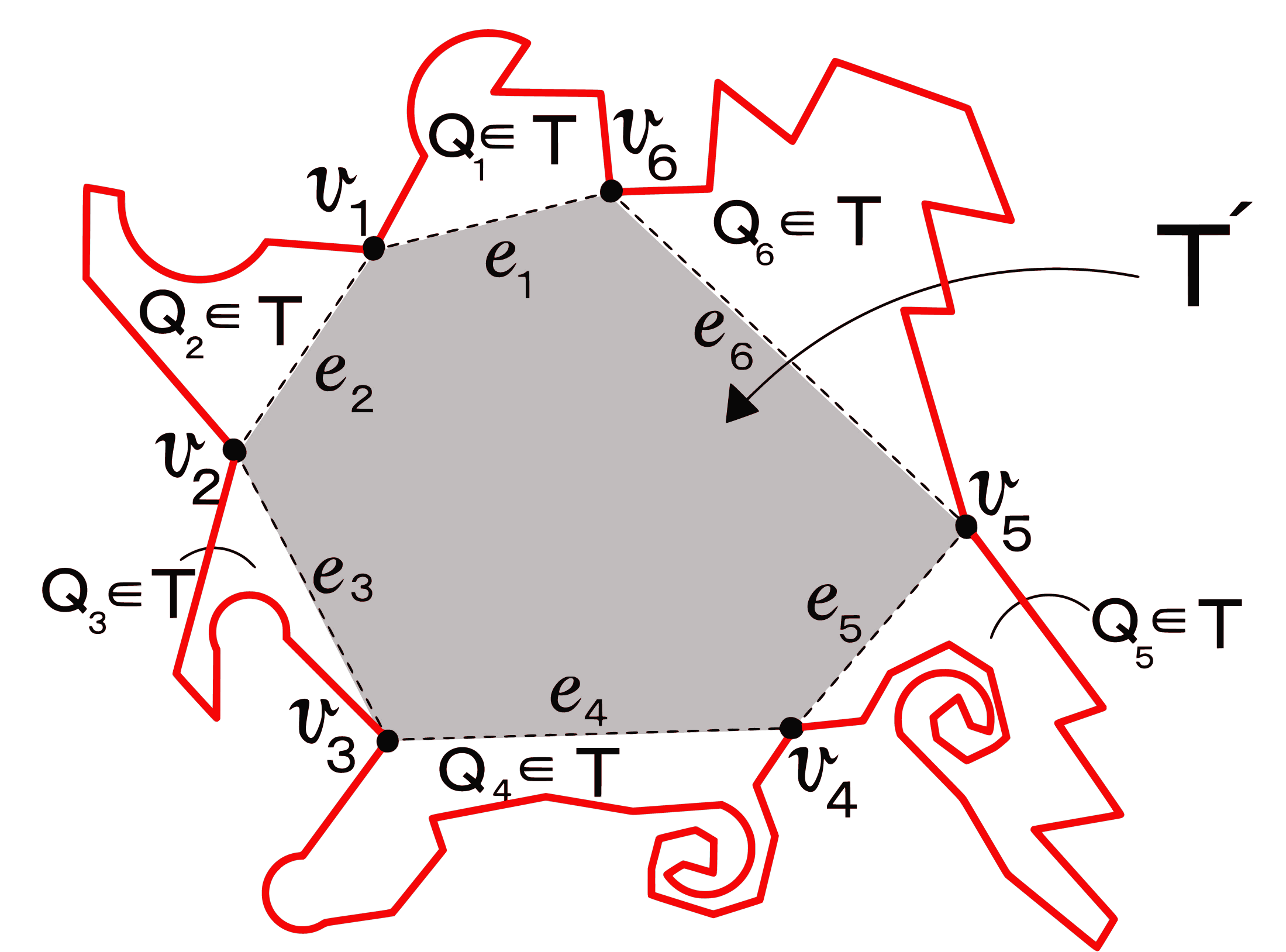} \\
    \vspace{1cm}
    {A $(T',T)$-chain of $Q$}\\
    \vspace{5mm}
   {$T$:}
     \includegraphics[width=\hsize]{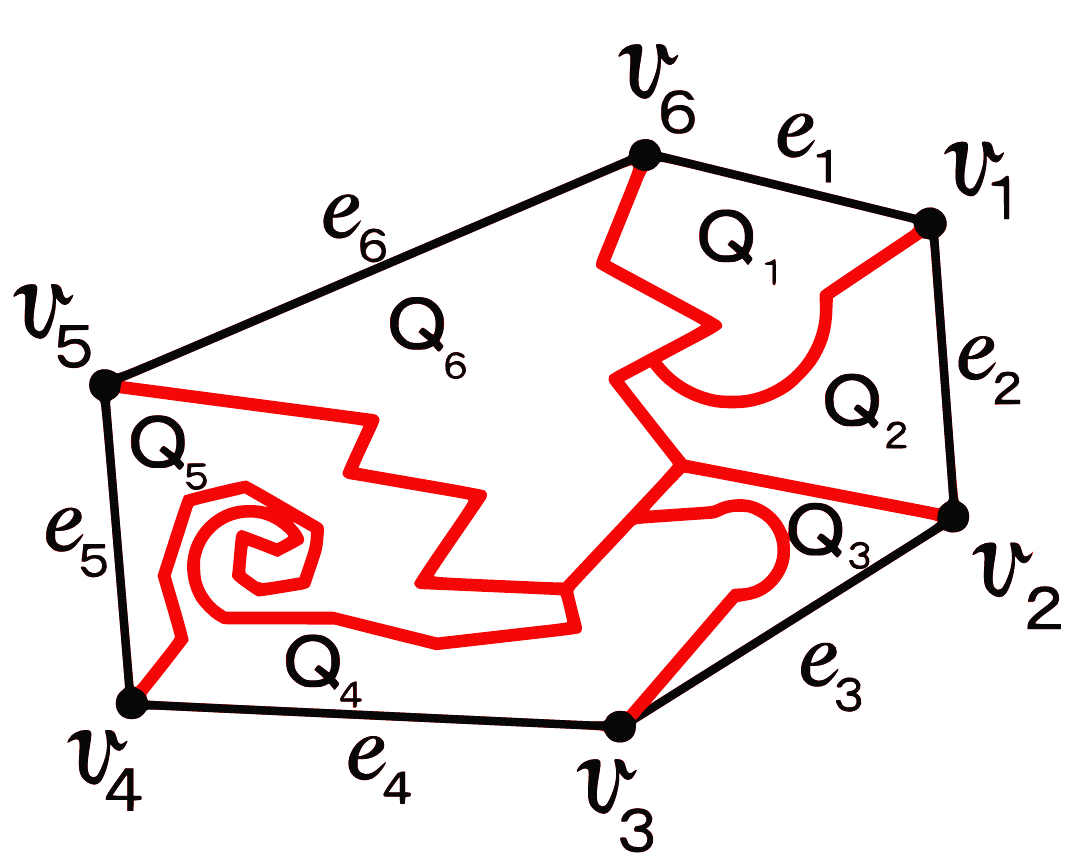} }
\caption{A $(T,T')$-chain of $P$ and a $(T',T)$-chain of $Q$.}\label{fig:chains}
\end{figure}

Combine a $(T, T')$-chain of $P$ with a $(T', T)$-chain of $Q$ such
that each segment of the perimeter, $e_i$, has a piece
$P_i'$ of $P$ on one side (right side) and a piece $Q_i$ of $Q$ on the other side (left side). The chain obtained in this manner is called a \emph{double chain of $(P, Q)$} (Fig.~\ref{fig:doublechain}).

\begin{figure}[h]
\parbox[t]{0.24\textwidth}{\centering ~\\
   \vspace{2cm}
    {$P$:}
     \includegraphics[width=\hsize]{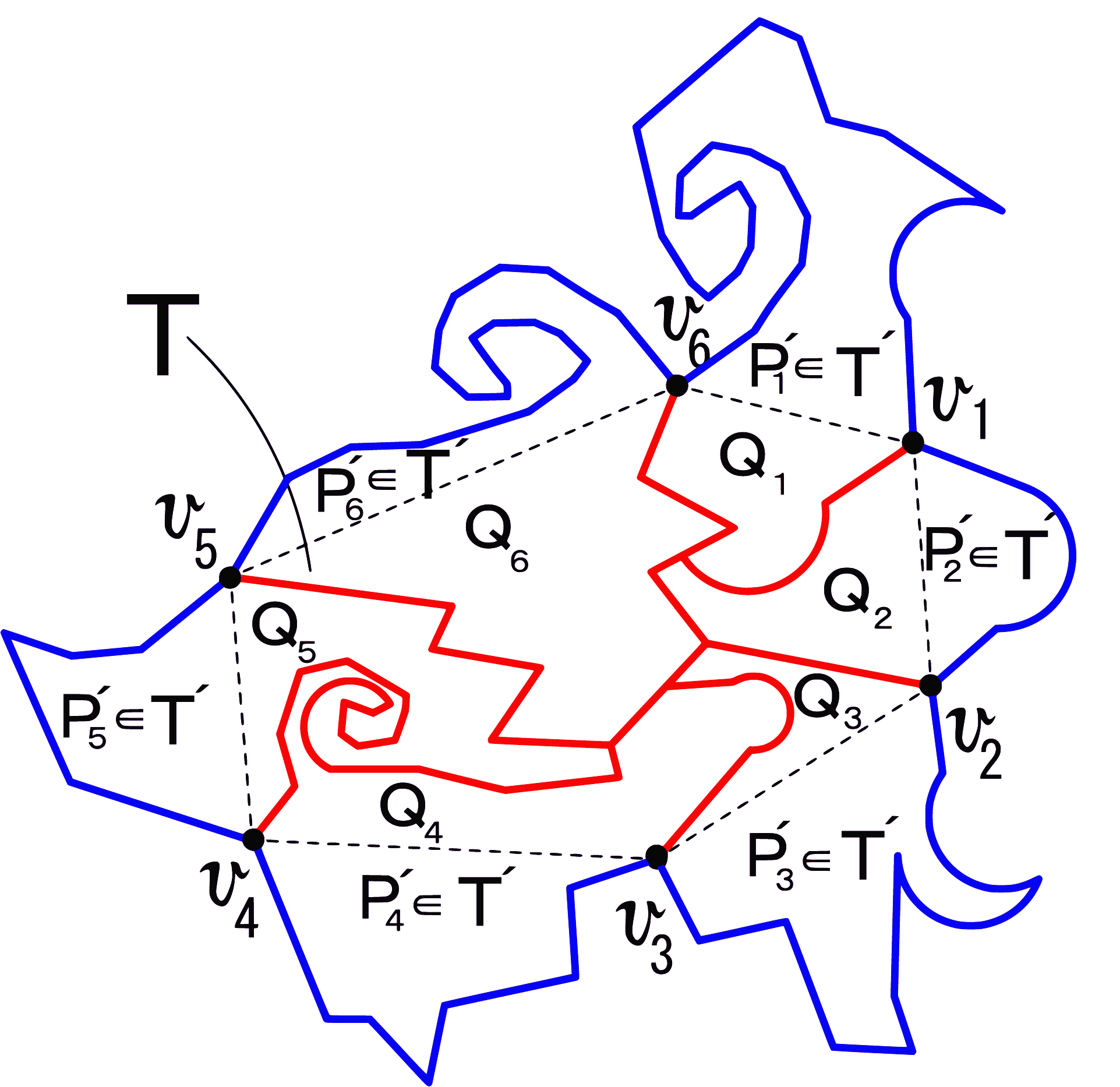} }
\parbox[t]{0.47\textwidth}{\centering ~\\
    \includegraphics[width=0.29\hsize]{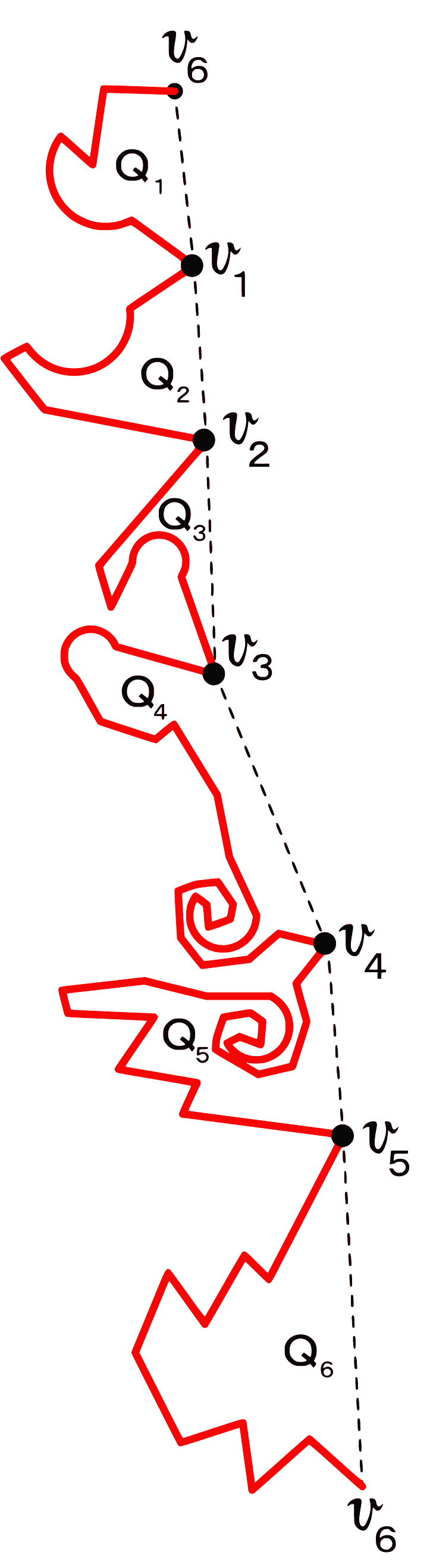}
    \includegraphics[width=0.44\hsize]{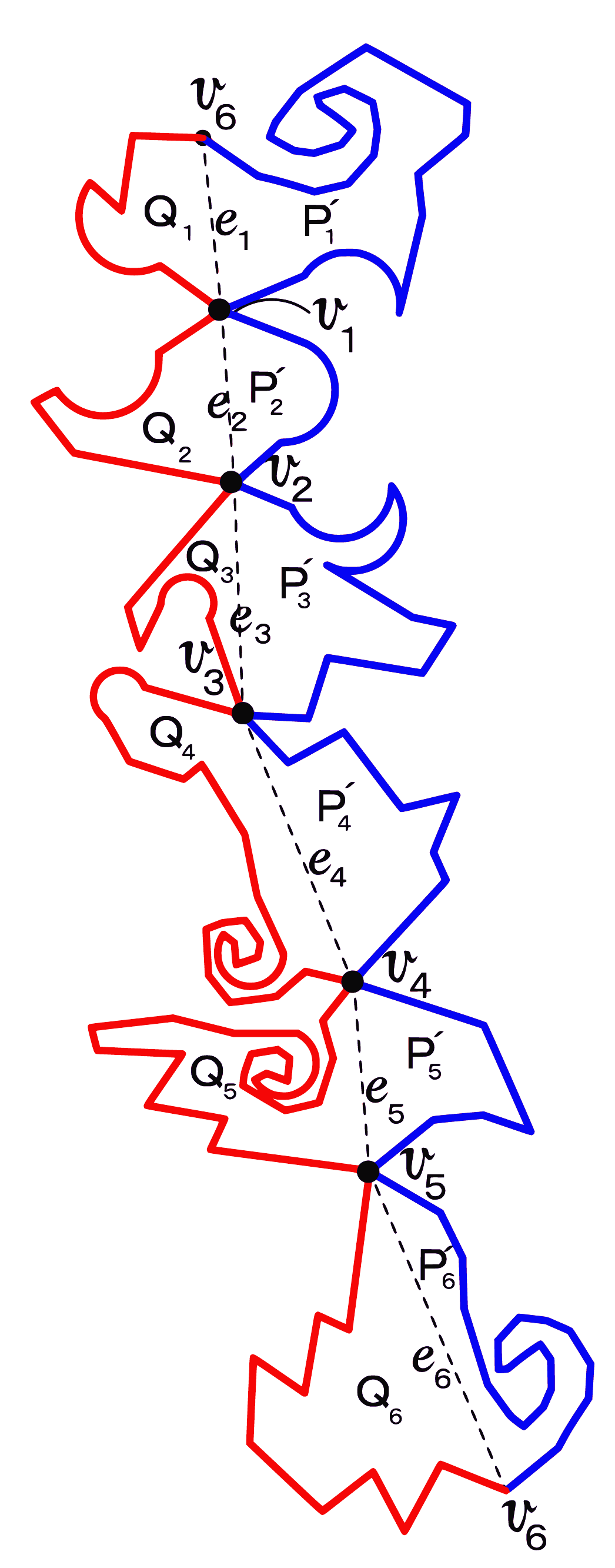}
    \includegraphics[width=0.21\hsize]{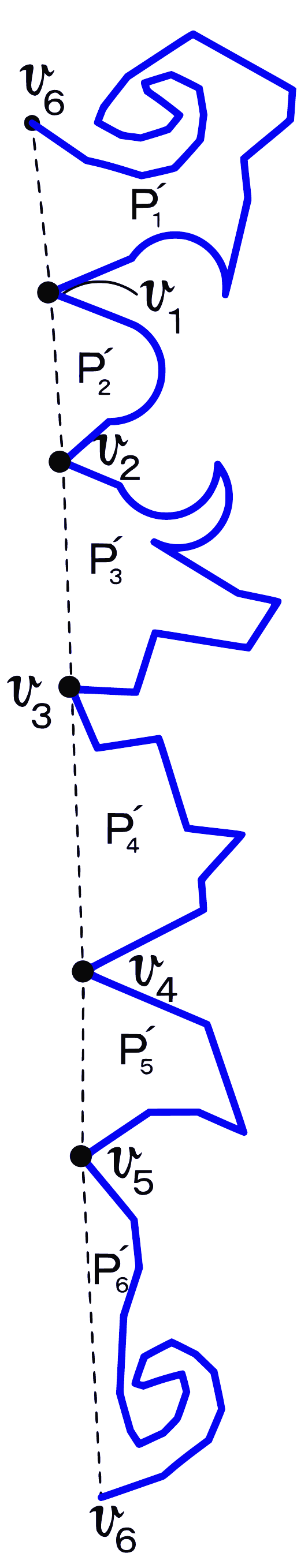}}
\parbox[t]{0.24\textwidth}{\centering ~\\
   \vspace{2cm}
    {$Q$:}
     \includegraphics[width=\hsize]{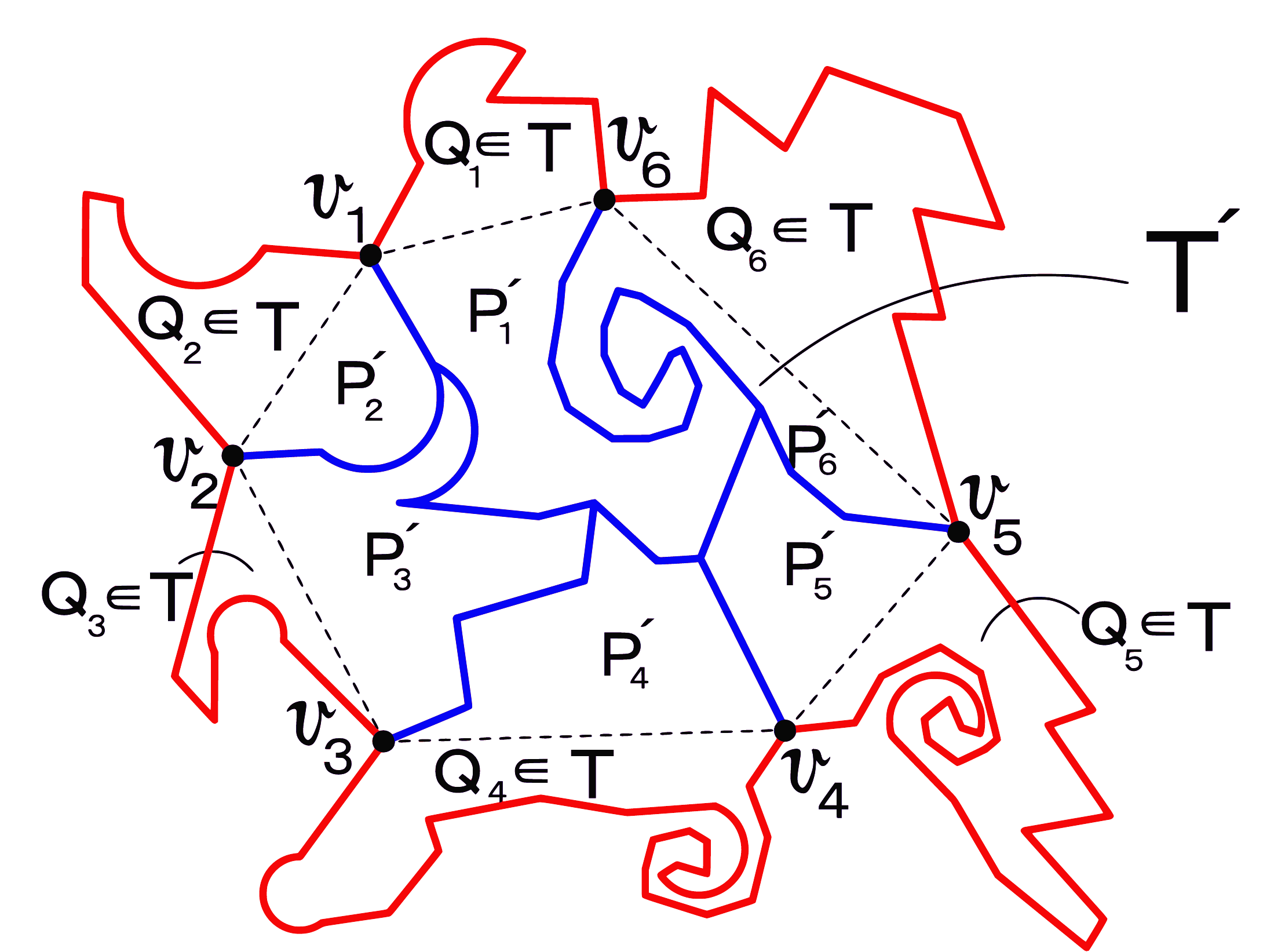} }
\caption{A double chain of $(P, Q)$.}\label{fig:doublechain}
\end{figure}

We say that a piece of a double chain is \emph{empty} if that piece consists of only a
perimeter part $e_i$. If a double chain has an empty piece, then we
distinguish one side of that edge from the other side so that it
satisfies the conditions for reversibility. If one of the end-pieces
(Say $P_1$ and $Q_1$ in Fig.~\ref{fig:doublechain}) of the double
chain of $(P, Q)$ is fixed and the remaining pieces are rotated
clockwise or counterclockwise, then 
figure $P$ and figure $Q$ are obtained respectively (Fig.~\ref{fig:doublechain}). The
following result is obtained from~\cite{r3}.

\begin{theorem}[Reversible Transformations Between Figures]  \label{thm:reversible}
Let $P$ be a figure with trunk $T$ and conjugate trunk $T'$, and let $Q$ have trunk
$T'$ and conjugate trunk $T$. Then $P$ is reversible to $Q$.
\end{theorem}

\paragraph{Remarks}
\begin{enumerate}
\item In Theorem~\ref{thm:reversible}, figure $P$ which is the union
  of $T$ and $n$ pieces $P_i'$ of the conjugate trunk $T'$ reversibly
  transforms into figure $Q$ which is the union of $T'$ and $n$
  pieces of $T$. 
\item “Harberdasher's puzzle” by H. Dudeney is also one such reversible
pair. In this puzzle, the figures $P$ and $Q$ are an
equilateral triangle and a square, respectively. The trunk $T$ and conjugate
trunk $T'$ are the identical parallelogram $T$ (the gray part in Fig.~\ref{fig:dudeney}).
\end{enumerate}

\section{Reversible nets of polyhedra}
A \emph{dissection tree} $D$ of a polyhedron $P$ is a tree drawn on
the surface of $P$ that spans all vertices of $P$. 
Cutting the surface of $P$ along $D$ results in a \emph{net} of
$P$. Notice that nets of some polyhedron $P$ may have
self-overlapping parts (Fig.~\ref{fig:overlap}). We allow such cases when discussing
reversible transformation of nets.  

\begin{figure}[h]
\parbox[t]{0.31\textwidth}{\centering ~\\
    \vspace{3mm}
     \includegraphics[scale=0.15]{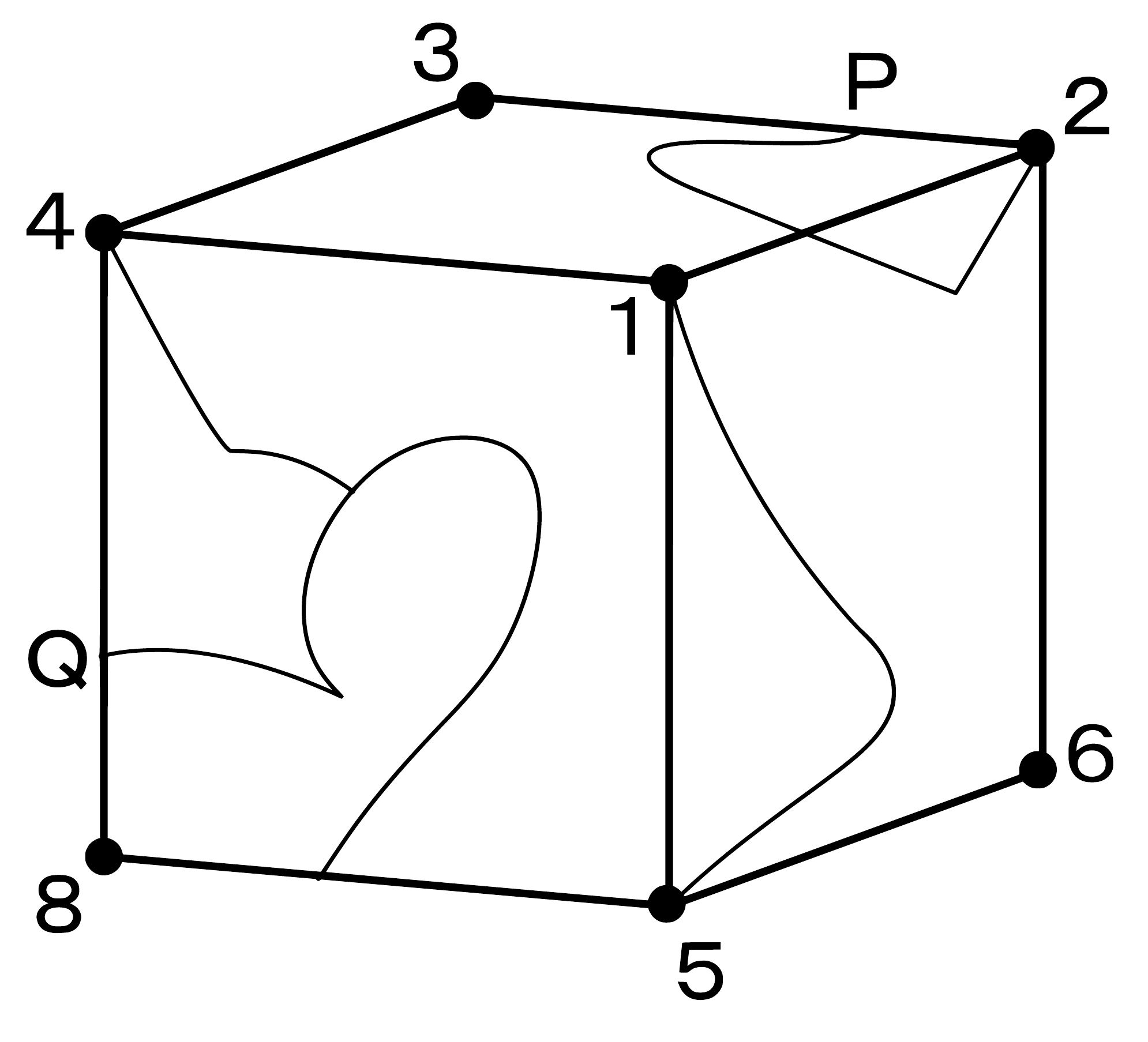} \\
    {front view} }
\parbox[t]{0.31\textwidth}{\centering ~\\
    \includegraphics[scale=0.18]{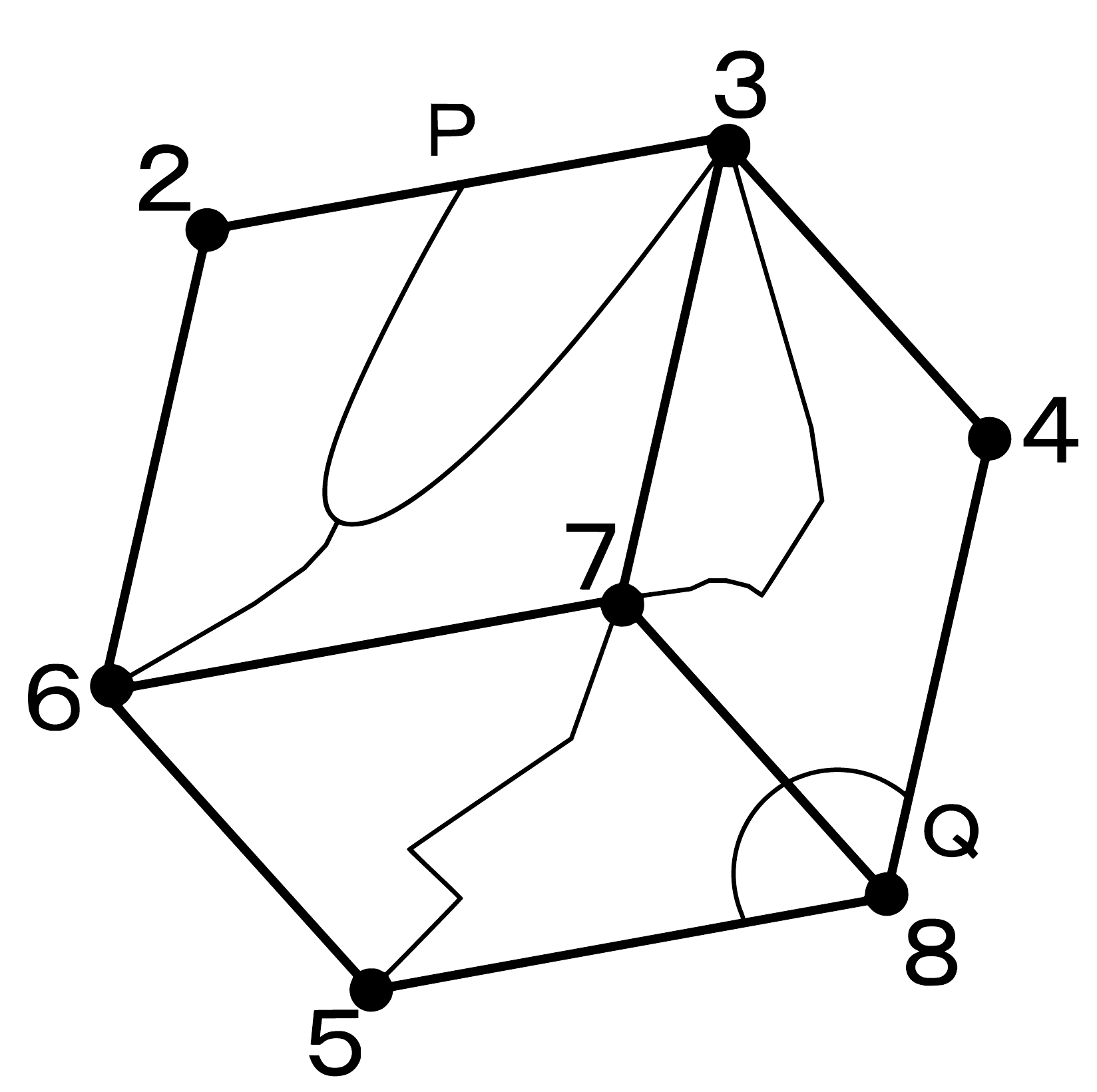}\\
    {back view} }
\parbox[t]{0.31\textwidth}{\centering ~\\
     \includegraphics[scale=0.7]{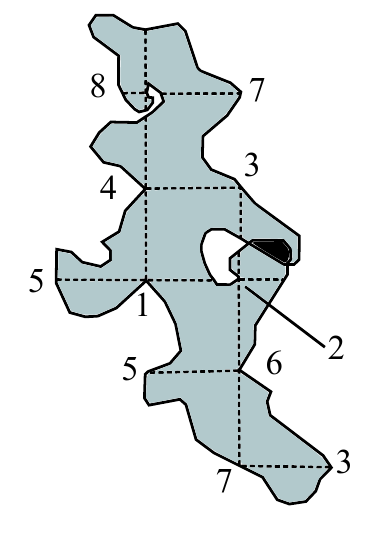} \\
    }
\caption{A net of a cube with self-overlapping part (the overlap is the black part).}\label{fig:overlap}
\end{figure}

\begin{theorem} \label{thm:reversible2}
Let $P$ be a polyhedron with $n$ vertices $v_1,\ldots,v_n$ and let $D_i$
($i=1,2$) be dissection trees on the surface of $P$. Denote by
$N_i$ ($i=1,2$) the nets of $P$ obtained by cutting $P$ along $D_i$
($i=1,2$), respectively. 
If  $D_1$ and $D_2$ don't properly cross,
then the pair of nets $N_1$ and $N_2$ is reversible, and has a double
chain composed of $n$ pieces. 
\end{theorem}

\begin{proof} 
Suppose that dissection trees $D_1$ (the red tree) and $D_2$ (the green
tree) on the surface of $P$ do not properly cross
(Fig.~\ref{fig:parcel}(a)). 
Then there exists a closed Jordan curve on the surface of $P$,
which separates the surface of $P$ into two pieces, one containing $D_1$, the
other containing $D_2$. 
Let $C$ be an arbitrary such curve (Fig.~\ref{fig:parcel}(b)). We
call $C$ a separating cycle. The net $N_1$, obtained by cutting $P$ along $D_1$,
contains an inscribed closed region $T$ whose boundary is $C$
(Fig.~\ref{fig:parcelnet}(a)). 
On the other hand, a net $N_2$ which is obtained by cutting $P$ along $D_2$
contains an inscribed conjugate region $T'$ whose boundary is the
opposite side of $C$ (Fig.~\ref{fig:parcelnet}(c)). Hence, a net $N_1$ has a trunk $T$ and a
conjugate trunk $T'$, and a net $N_2$ has a trunk $T'$ and a conjugate trunk
$T$. By Theorem~\ref{thm:reversible} this pair of $N_1$ and $N_2$ is reversible (Fig.~\ref{fig:parcelnet}(b)). \qed
\end{proof}

\begin{figure}[h]
\parbox[t]{0.49\textwidth}{\centering ~\\
     \includegraphics[scale=0.12]{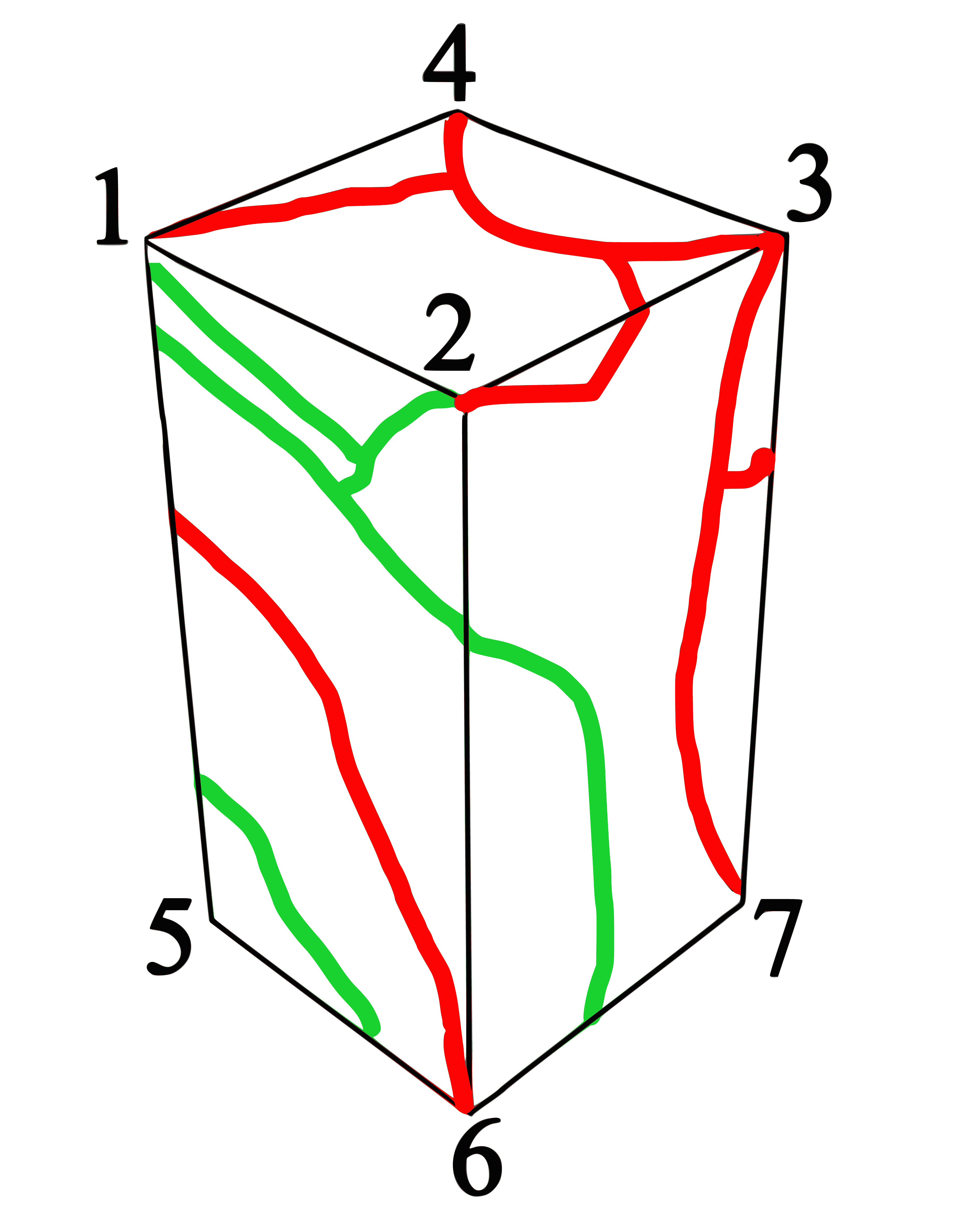} \\
    {(a)} }
\parbox[t]{0.49\textwidth}{\centering ~\\
    \includegraphics[scale=0.08]{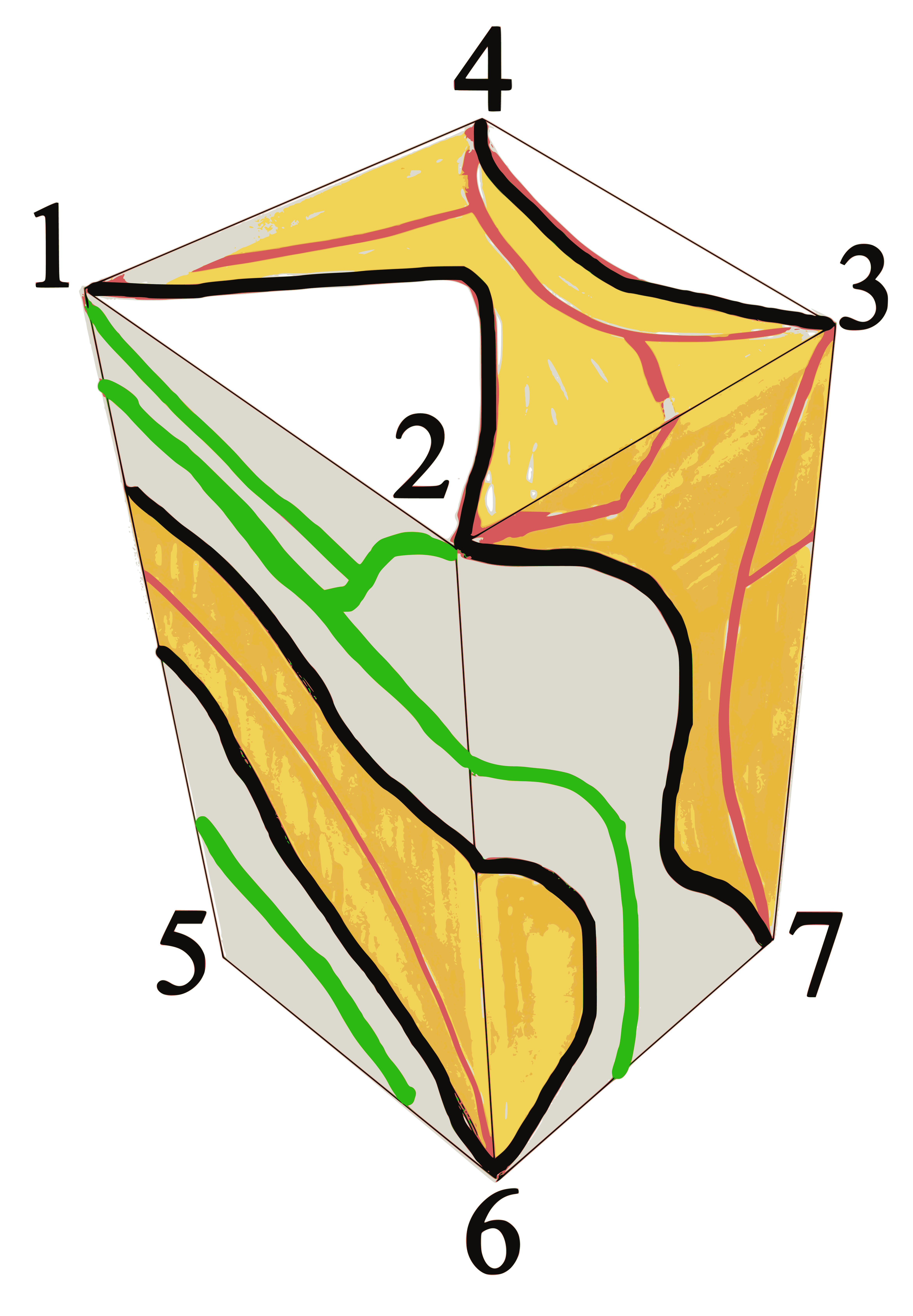}\\
    {(b)} }
\caption{A polyhedron $P$ with dissection trees $D_1$ (red tree) and $D_2$ (green tree),
a separating cycle $C$ (black cycle).}\label{fig:parcel}
\end{figure}

\begin{figure}[h]
\parbox[b]{0.25\textwidth}{\centering ~\\
     \includegraphics[scale=0.04]{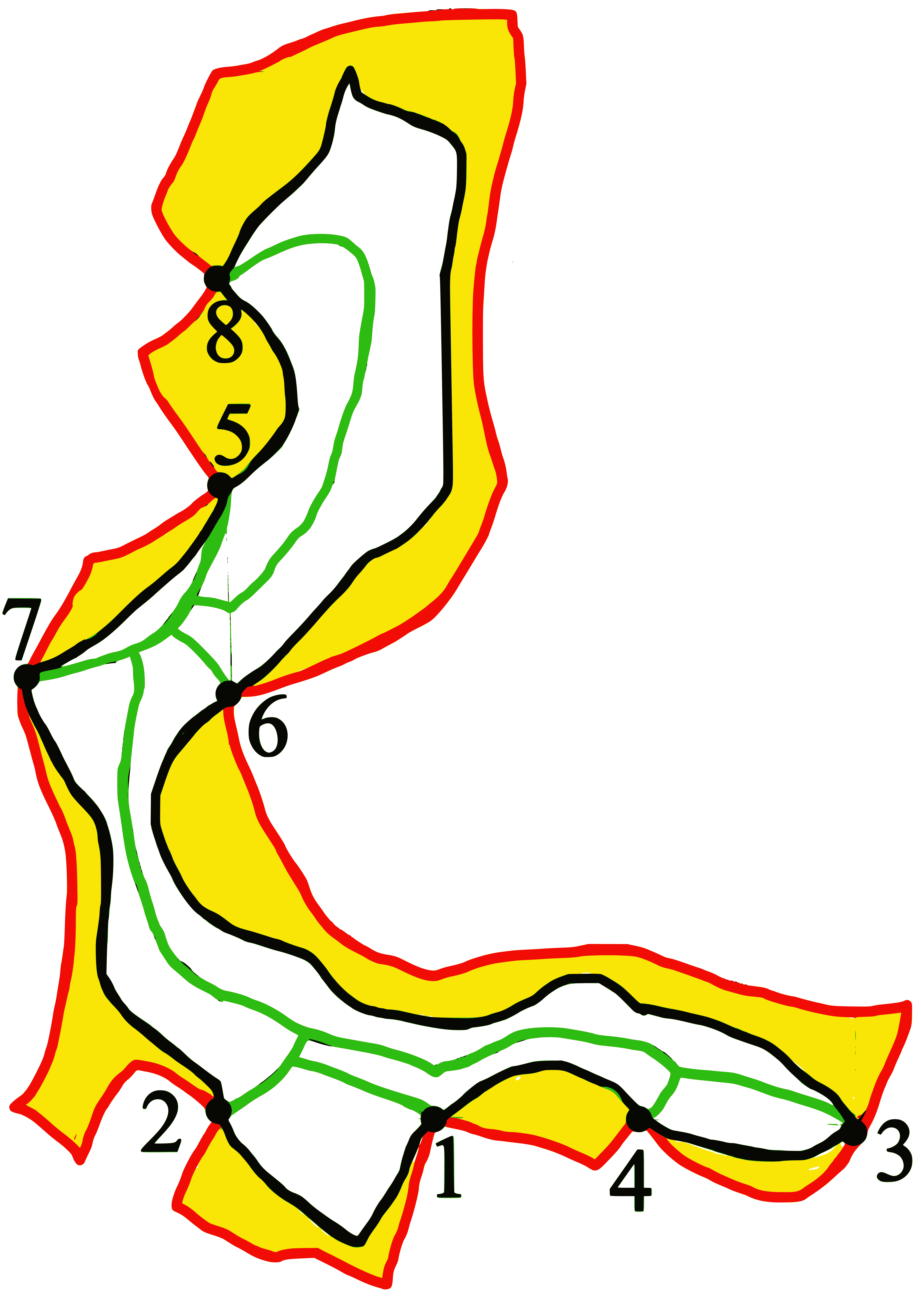} \\
    {(a)} }
\parbox[b]{0.37\textwidth}{\centering ~\\
    \includegraphics[scale=0.06]{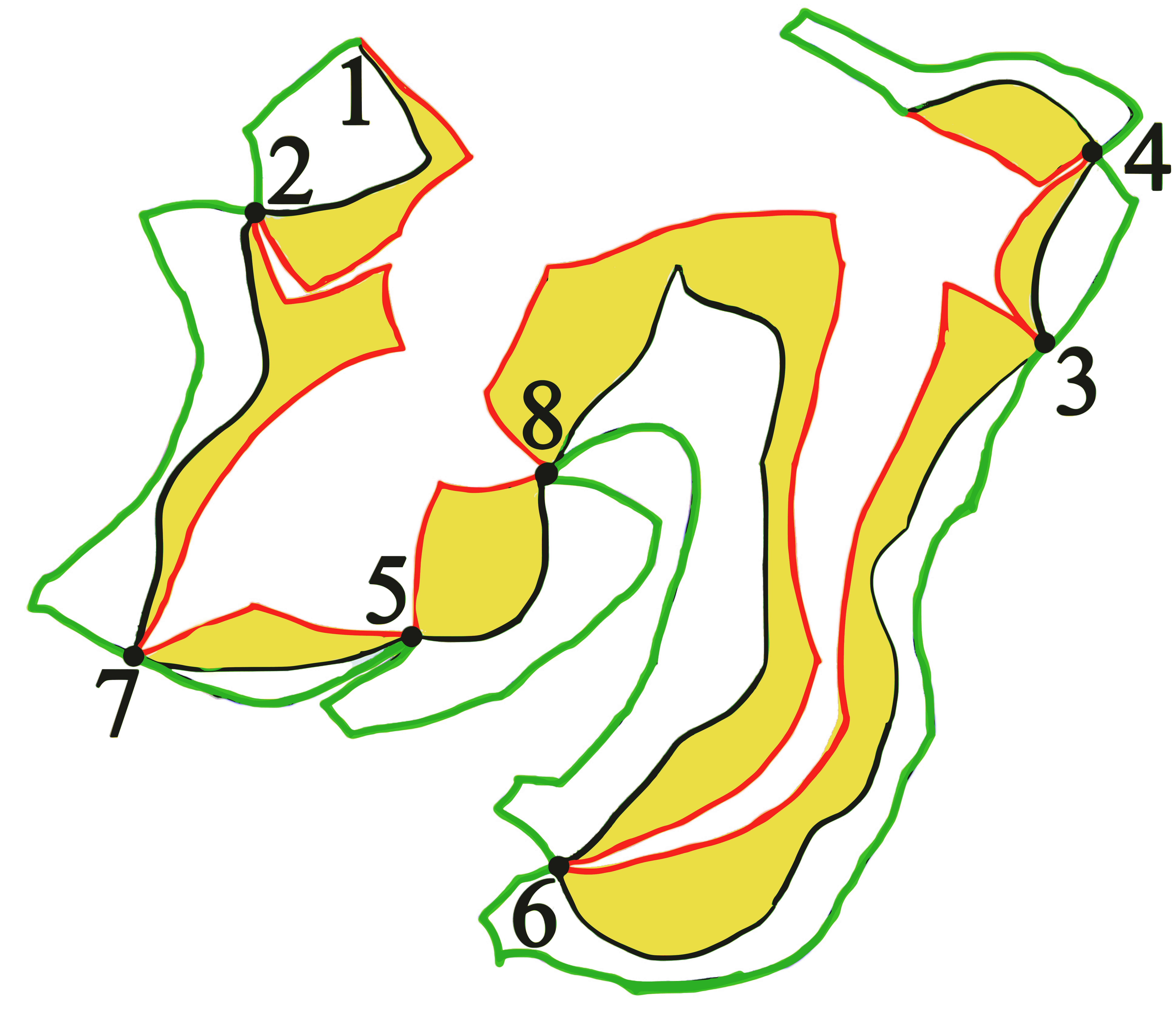}\\
    {(b)} }
\parbox[b]{0.31\textwidth}{\centering ~\\
     \includegraphics[scale=0.04]{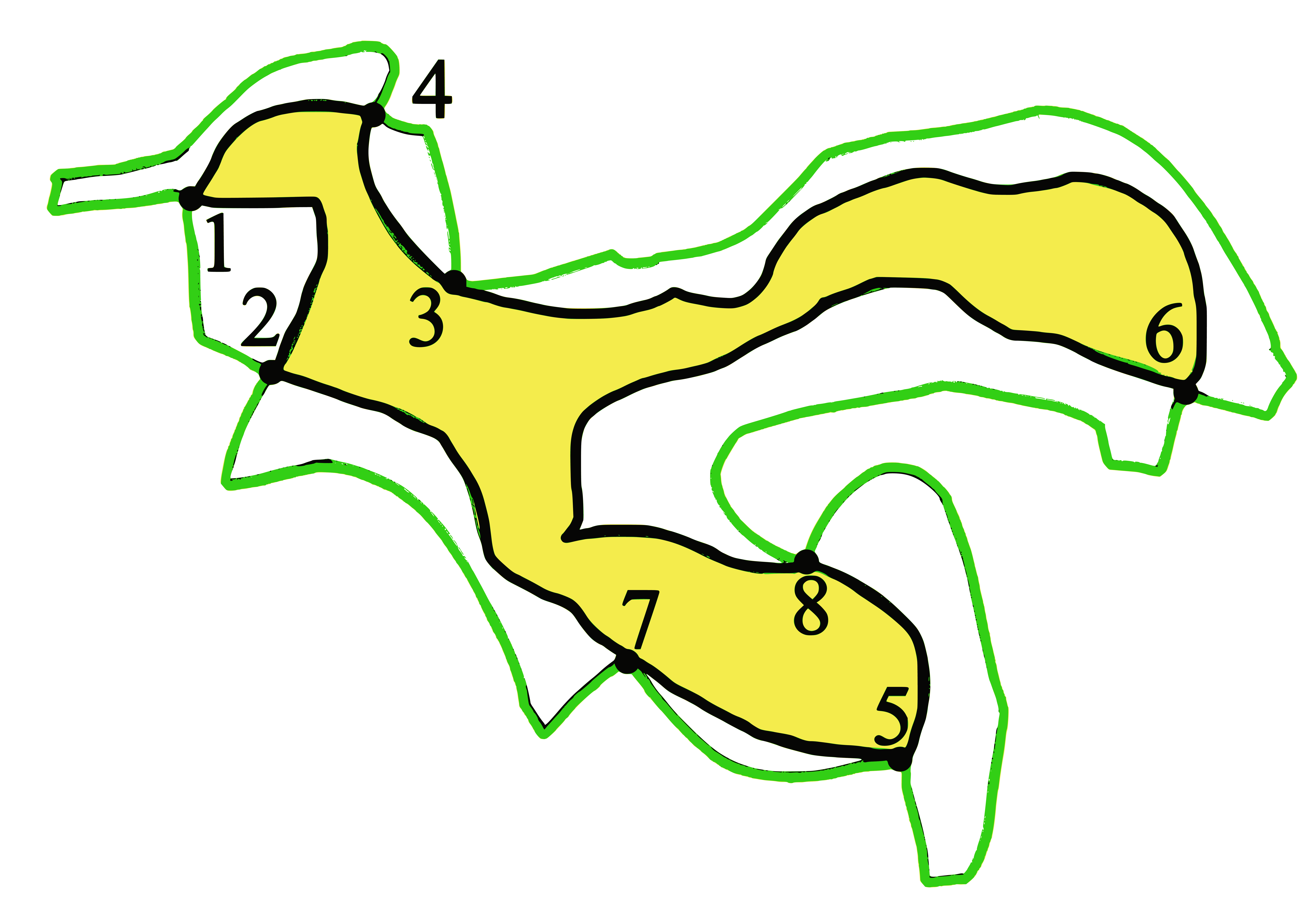} \\
    {(c)}}
\caption{Nets $N_1$ and $N_2$ obtained by cutting the surface of $P$
  along $D_1$ and $D_2$, respectively.}\label{fig:parcelnet}
\end{figure}

\begin{theorem}\label{thm:manynets}
For any net $N_1$ of a polyhedron $P$ with $n$ vertices, there exist infinitely many nets $N_2$ of $P$ such that $N_1$ is reversible to $N_2$.
\end{theorem}
 
\begin{proof}  
Any net $N$ of $P$ has a one-to-one correspondence with a dissection
tree $D$ on the surface of $P$. Let the dissection tree of $N_i$ be
$D_i$ ($i=1,2$), respectively (Fig.~\ref{fig:swirl}(a)). 
The perimeter of $N_i$ can be decomposed into several parts in which
each is congruent to an edge of $D_i$. 
Moreover, a vertex with degree $k$ on $D_i$ appears $k$ times on the
perimeter of $N_i$. These duplicated vertices of $v_i$ are labeled as
$v_i',v_i'',\ldots$. 

Choose an arbitrary vertex $v_k$ among $v_k,v_k',v_k'',\ldots$ on
$N_1$ as a representative and denote it by $v_k^*$, where $k = 1,
2,\ldots, n$.  
Since $N_1$ is connected, it is possible to
draw infinitely many arbitrary spanning trees $D_2$, each of which connects $v_k^*$ ($k = 1,
2,\ldots, n$) inside $N_1$ (Fig.~\ref{fig:swirl}(b)). 
Then, any such $D_2$ doesn't intersect $D_1$. %other than at vertices of $P$ on
%the surface of $P$ 
(Fig.~\ref{fig:swirl}(c)). As in
Theorem~\ref{thm:reversible2}, dissect $N_1$ along $D_2$ into $n$ 
pieces $P_1,\ldots,P_n$ , and then connect them in sequence using
$n-1$ hinges on the perimeter of $N_1$ to form a chain. 
Fix one of the end-pieces of the chain and rotate the
remaining pieces then forming net $N_2$ which is obtained by cutting $P$
along $D_2$ (Fig.~\ref{fig:swirl} (d)).  \qed
\end{proof}

\begin{figure}[h]
\parbox[t]{\textwidth}{
  \noindent
  \parbox[t]{3.5mm}{\vspace{0mm}(a)}
  \parbox[t]{0.26\textwidth}{\vspace{0mm}\includegraphics[width=\hsize]{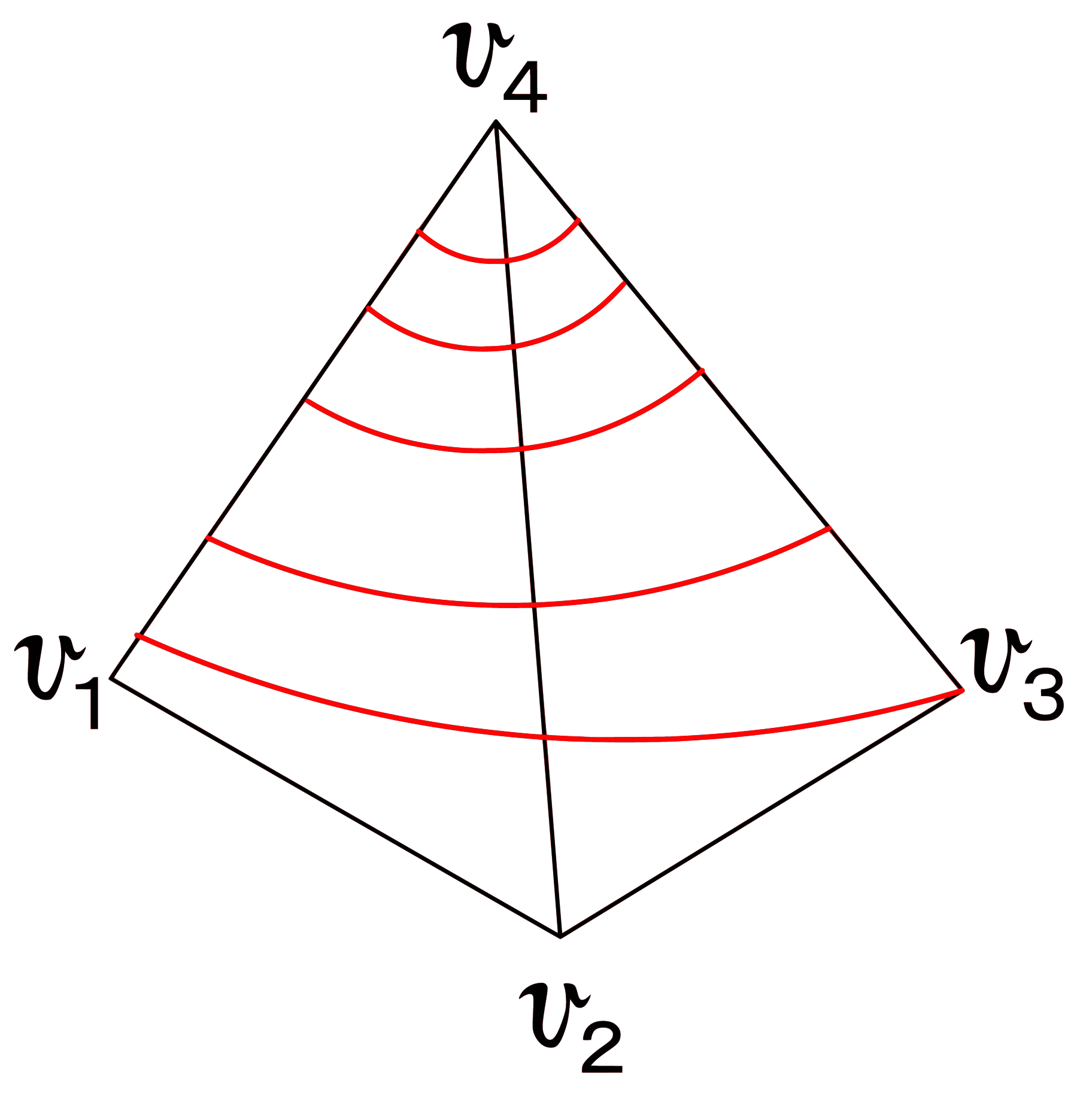}}
  %\parbox[t]{1.6mm}{\vspace{13mm}+}
  \parbox[t]{0.26\textwidth}{\vspace{0mm}\includegraphics[width=\hsize]{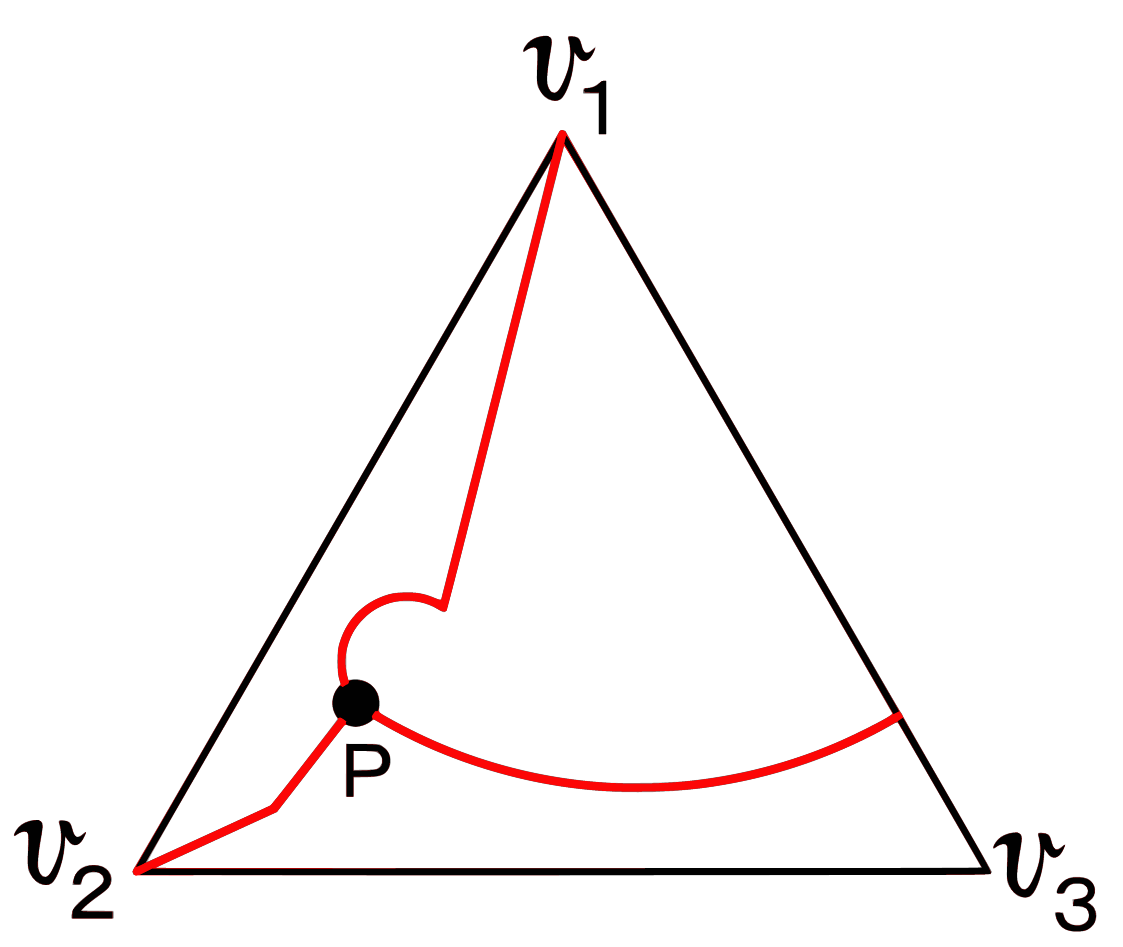}}
  %\parbox[t]{1.6mm}{\vspace{13mm}=}
  \parbox[t]{0.47\textwidth}{\vspace{0mm}\includegraphics[width=\hsize]{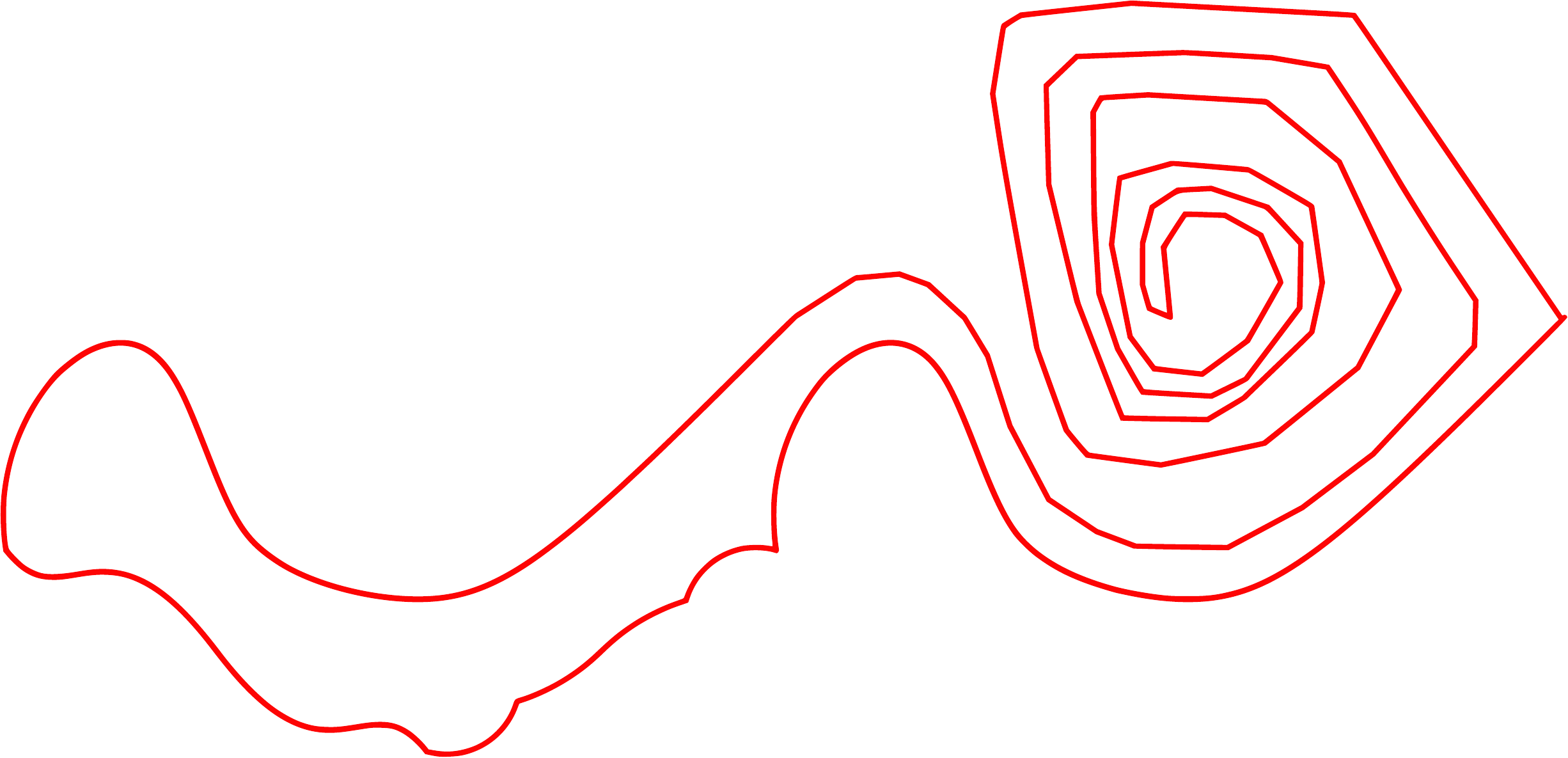}}\\
\begin{picture}(0,0)
\put(20,70){\mbox{$D_1$:}}
\put(200,70){\mbox{$N_1$:}}
\put(100,50){\mbox{$+$}}
\put(180,50){\mbox{$\Rightarrow$}}
\end{picture}
  \mbox{\parbox[t]{3.5mm}{\vspace{0mm}(b)}
  \parbox[t]{0.46\textwidth}{\vspace{0mm}\includegraphics[width=\hsize]{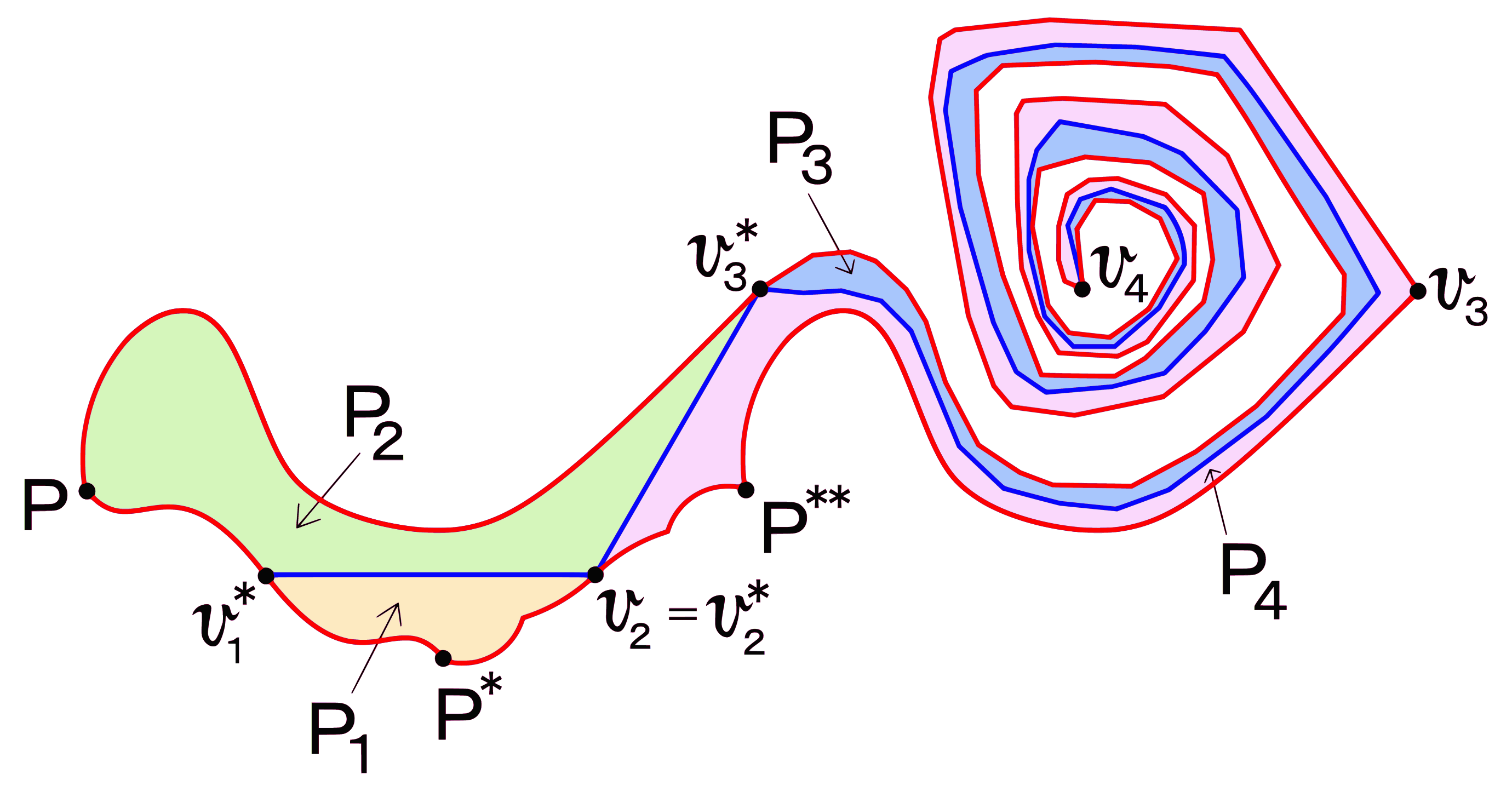}}
  \parbox[t]{0.26\textwidth}{\vspace{0mm}\includegraphics[width=\hsize]{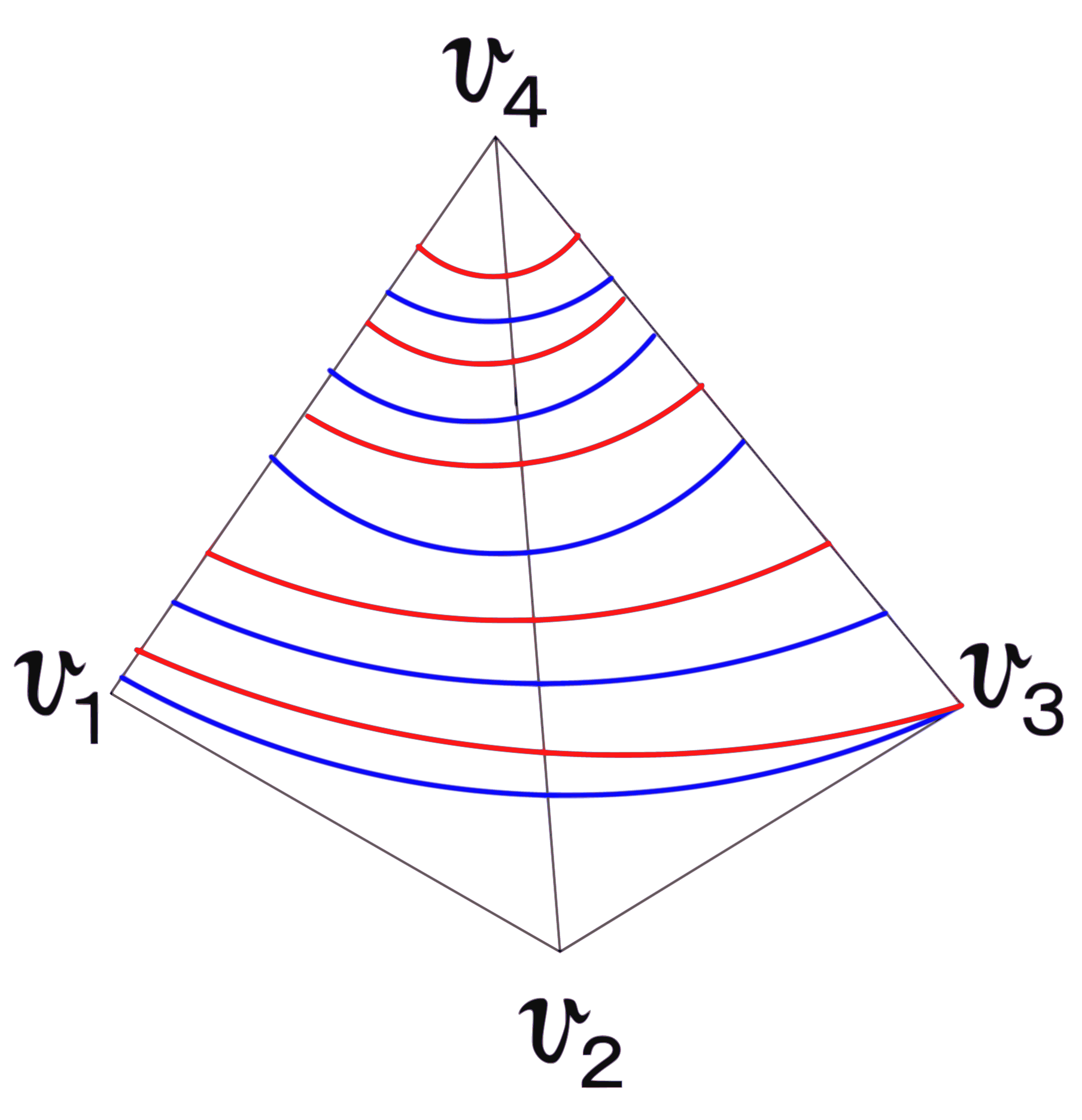}}
  \parbox[t]{0.26\textwidth}{\vspace{0mm}\includegraphics[width=\hsize]{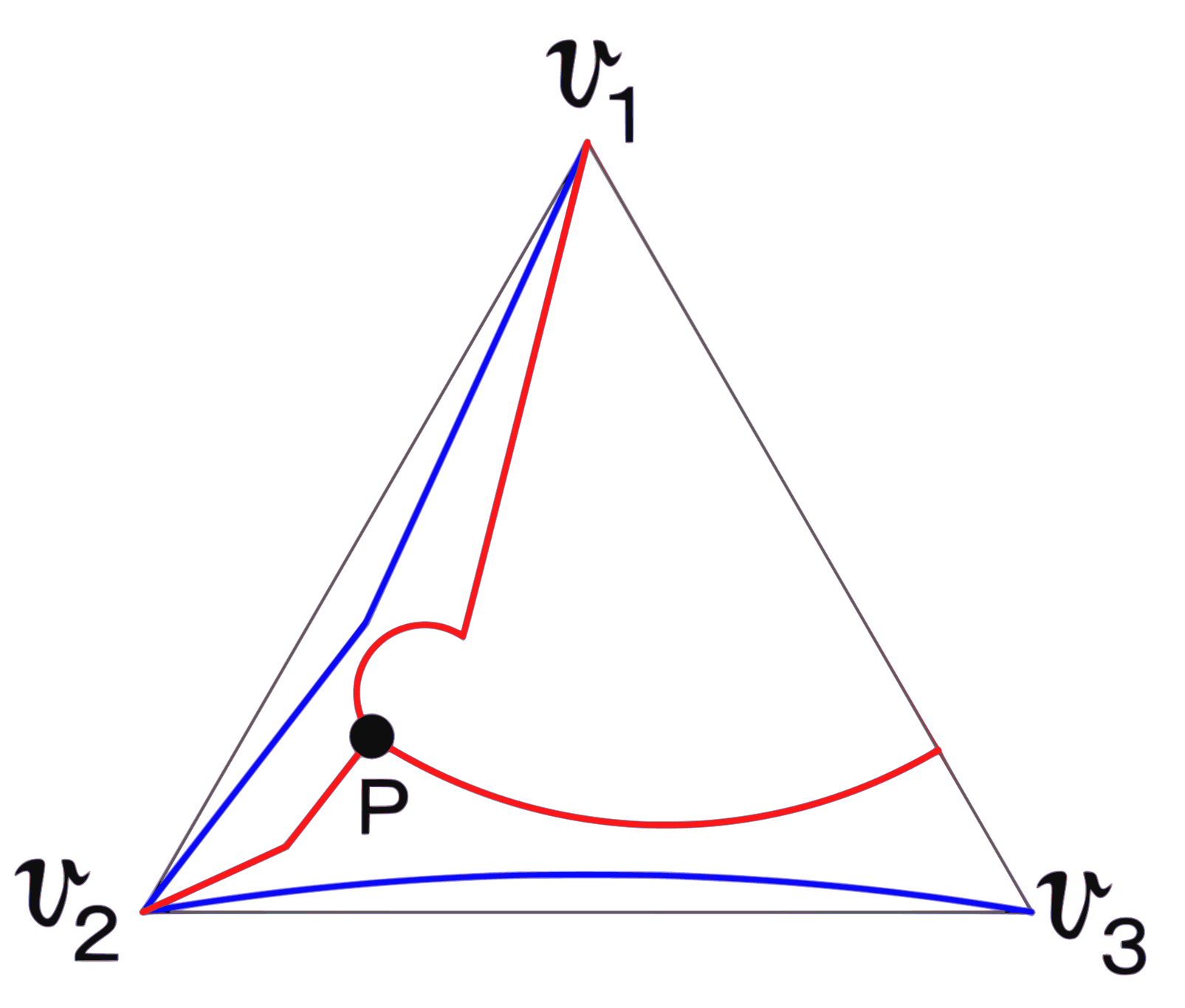}}}
  \\
  \parbox[t]{3.5mm}{\vspace{0mm}(c)}
  \parbox[t]{0.37\textwidth}{\vspace{0mm}\includegraphics[width=\hsize]{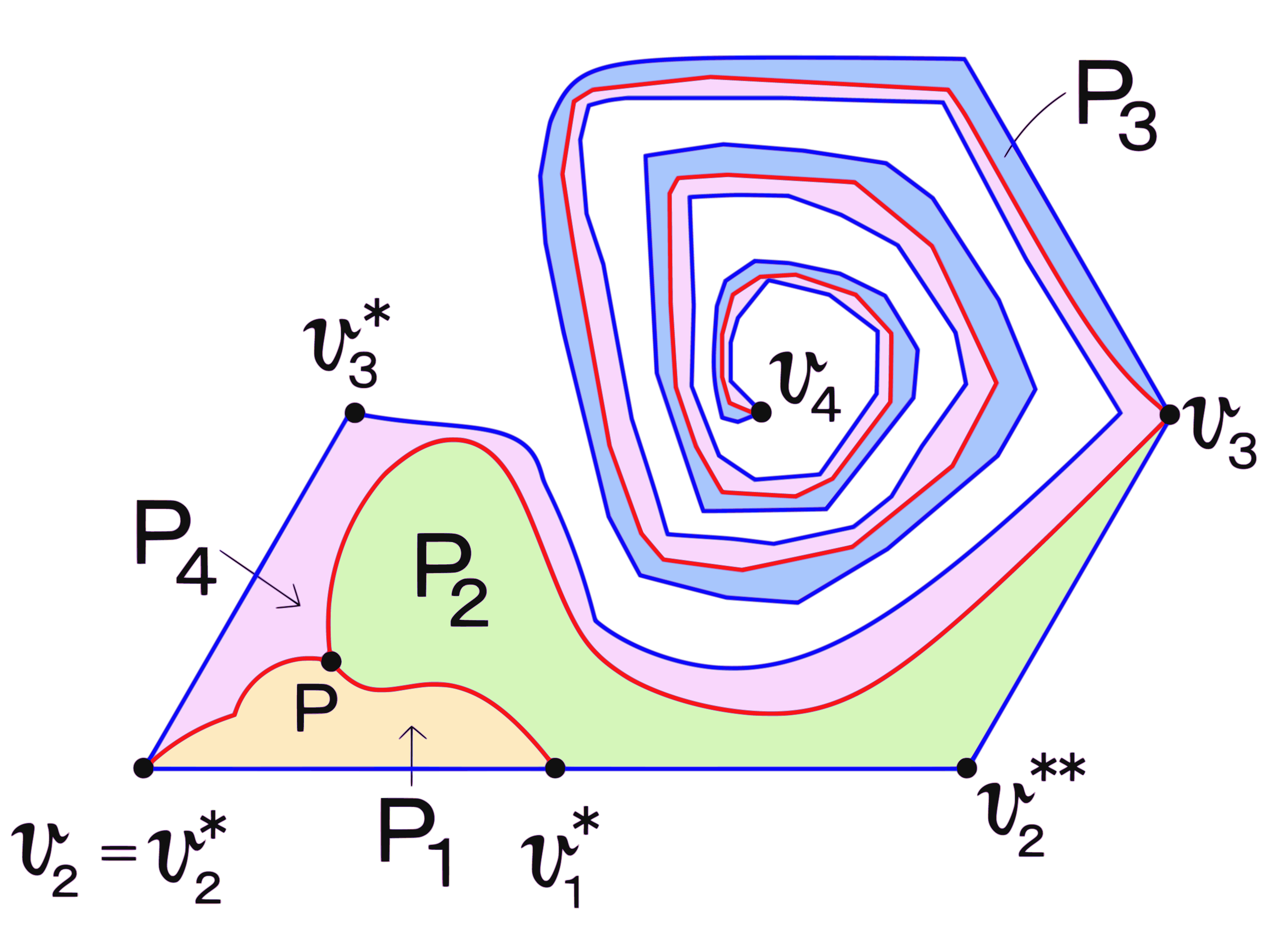}}
}
\caption{A swirl net of a regular tetrahedron.}\label{fig:swirl}
\end{figure}
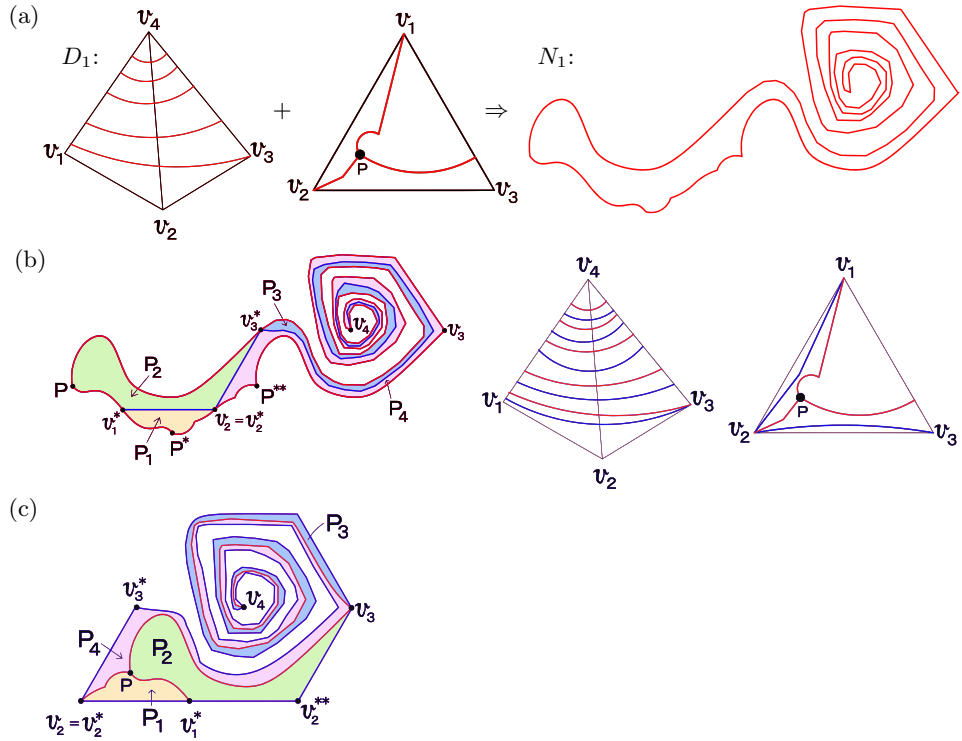

\begin{corollary}[Envelope magic \cite{r7}]  
Let $E$ be an arbitrary doubly covered polygon (dihedron) and let
$D_1$ and $D_2$,  be dissection trees of $E$. If dissection tree $D_1$
doesn't properly cross dissection tree $D_2$,
then a pair of nets $N_1$ and $N_2$ obtained by cutting the surface of
$E$ along $D_1$ and $D_2$ is reversible (Fig.~\ref{fig:lobster}).  
\end{corollary}

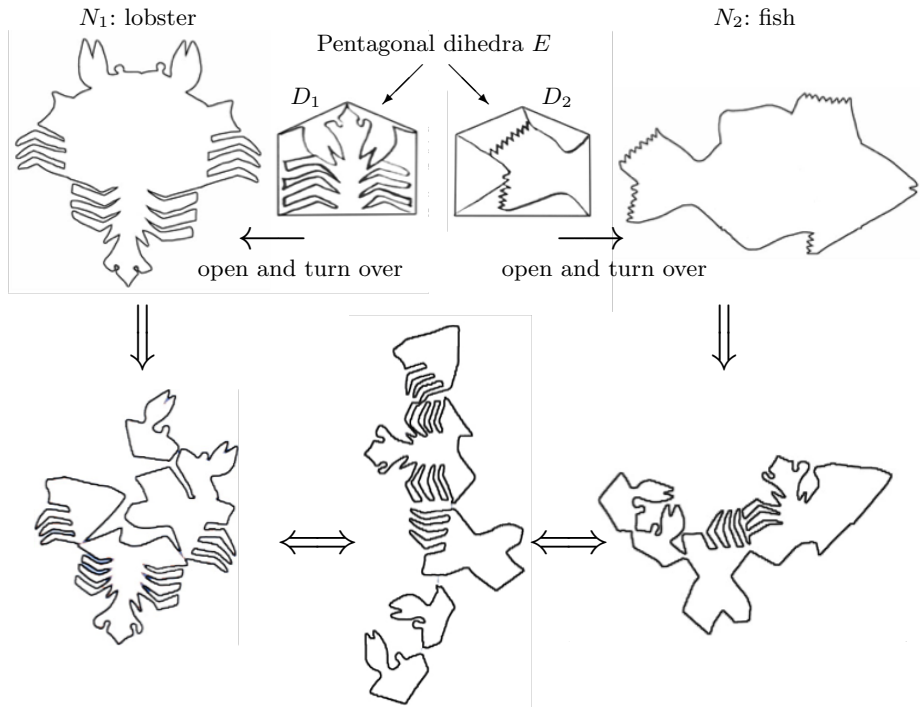
\begin{figure}[h]
\parbox[t]{\hsize}{
  \noindent\mbox{
  \parbox[t]{0.46\textwidth}{\vspace{0mm}\includegraphics[width=\hsize]{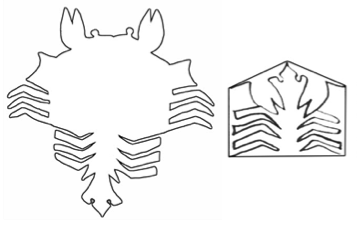}}
  \parbox[t]{0.52\textwidth}{\vspace{0mm}
   \mbox{
    \parbox[t]{0.33\hsize}{\vspace{8mm}\includegraphics[width=\hsize]{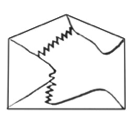}}
    \parbox[t]{0.66\hsize}{\vspace{0mm}\includegraphics[width=\hsize]{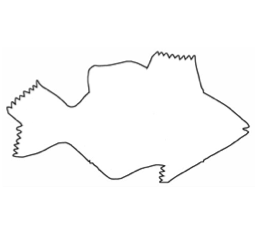}}
    }
    }}
   \mbox{
  \parbox[t]{0.28\textwidth}{\vspace{10mm}\centering\includegraphics[width=0.8\hsize]{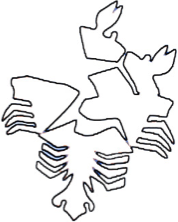}}
  \hspace{6mm}
  \parbox[t]{0.25\textwidth}{\vspace{0mm}\centering\includegraphics[width=0.8\hsize]{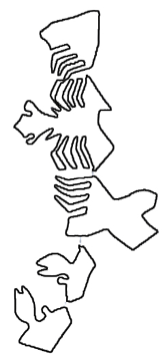}}
  \hspace{-6mm}
  \parbox[t]{0.46\textwidth}{\vspace{10mm}\centering\includegraphics[width=0.8\hsize]{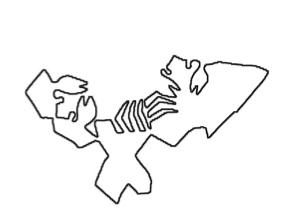}}
}
\begin{picture}(0,0)
\put(30,260){\mbox{$N_1$: lobster}}
\put(270,260){\mbox{$N_2$: fish}}
\put(120,250){\mbox{Pentagonal dihedra $E$}}
\put(170,245){\vector(1,-1){15}}
\put(160,245){\vector(-1,-1){15}}
\put(110,230){\mbox{$D_1$}}
\put(205,230){\mbox{$D_2$}}
\put(90,175){\mbox{\scalebox{1.8}{$\longleftarrow$}}}
\put(75,165){\mbox{open and turn over}}
\put(190,165){\mbox{open and turn over}}
\put(210,175){\mbox{\scalebox{1.8}{$\longrightarrow$}}}
\put(50,155){\mbox{\rotatebox{-90}{\scalebox{1.8}{$\Longrightarrow$}}}}
\put(270,155){\mbox{\rotatebox{-90}{\scalebox{1.8}{$\Longrightarrow$}}}}
\put(105,60){\mbox{\scalebox{1.8}{$\Longleftrightarrow$}}}
\put(200,60){\mbox{{\scalebox{1.8}{$\Longleftrightarrow$}}}}
\end{picture}
}
\caption{A lobster transforms into a fish ; The separating cycle C is the hem of a pentagonal dihedron.}\label{fig:lobster}
\end{figure}

The previous two theorems show that it is always possible to dissect
any polyhedron $P$ into two nets that are reversible, however, as
mentioned in the beginning of this section, those nets may sometimes
self-overlap when embedded in the plane. One may then ask whether a
convex polyhedron $P$ always has a pair of reversible non self-overlapping
nets. The following theorem answers in the positive. 

\begin{theorem}  
For any convex polyhedron $P$, there exists an infinity of pairs of non
self-overlapping nets of $P$ that are reversible.
\end{theorem}

\begin{proof}  
Choose an arbitrary point $s$ on the surface of $P$, but not on a
vertex. The \emph{cut locus} of $s$ is the set of all points $t$ on
the surface of $P$ such that the shortest path from $s$ to $t$ is not
unique. It is well known that the cut locus of $s$ is a tree that
spans all vertices of $P$. Cutting $P$ along the cut locus produces
the \emph{source unfolding}, which does not overlap~\cite{r10}. 
Let $D_1$ be the cut locus from $s$, and $N_1$ the corresponding non
self-overlapping net. The net $N_1$ is a star-shaped polygon, and the
shortest path from $s$ to any point $t$ in $P$ unfolds to a straight
line segment contained in $N_1$. 
The dissection tree $D_2$ is constructed by cutting $P$ along the
shortest path from $s$ to every vertex of $P$. The net $N_2$ thus
produced is a \emph{star unfolding} and also does not
overlap~\cite{r8}. Note also that the shortest path from $s$ to any
vertex of $P$, when cutting the source tree $D_1$, unfolds to a
straight line segment from $s$ to the corresponding vertex on
$N_1$. Therefore $D_1$ and $D_2$ do not properly intersect (In fact
$D_1$ and $D_2$ may coincide but not properly cross. In order to avoid
this, it suffices to choose $s$ not on the cut locus of any vertex of
$P$.)  
By Theorem~\ref{thm:reversible2}, $N_1$ and $N_2$ are reversible.    \qed
\end{proof}

\section{Reversibility and Tessellability for Nets of An
  Isotetrahedron}
A tetrahedron $T$ is called an \emph{isotetrahedron} if all faces of
$T$ are congruent. Note that there are infinitely many non-similar
isotetrahedra. Every net of an isotetrahedron tiles the
plane~\cite{r2}. Moreover, all nets of isotetrahedron can be
topologically classified into five types~\cite{r3}. 
By Theorem~\ref{thm:reversible2} and Theorem~\ref{thm:manynets}, the
following theorem is obtained:

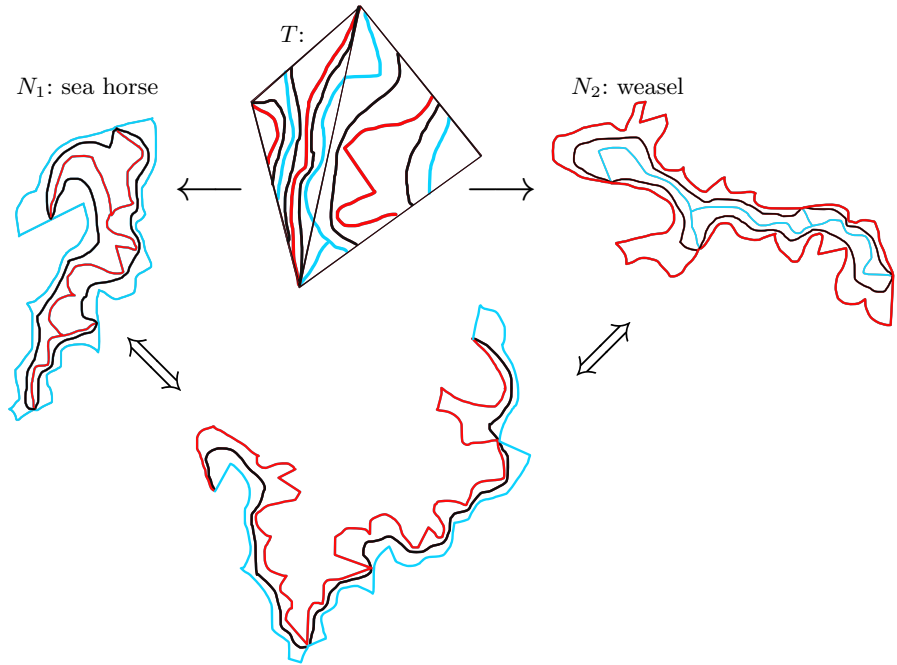
\begin{figure}[h]
\parbox[t]{\hsize}{
  \noindent
  \parbox[t]{0.21\textwidth}{\vspace{15mm}\includegraphics[width=\hsize]{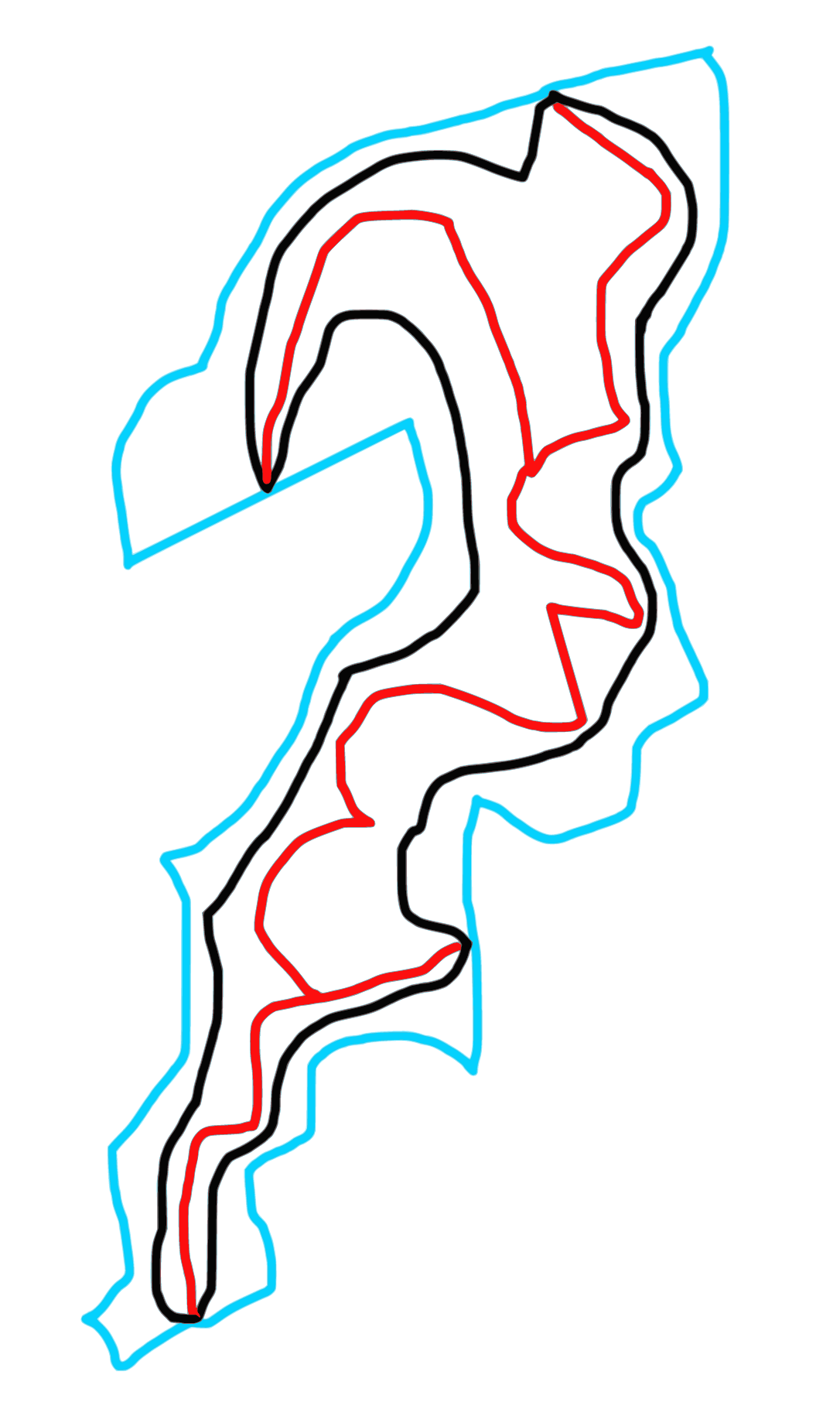}}
  \parbox[t]{0.38\textwidth}{\vspace{0mm}\centering
      \includegraphics[width=0.7\hsize]{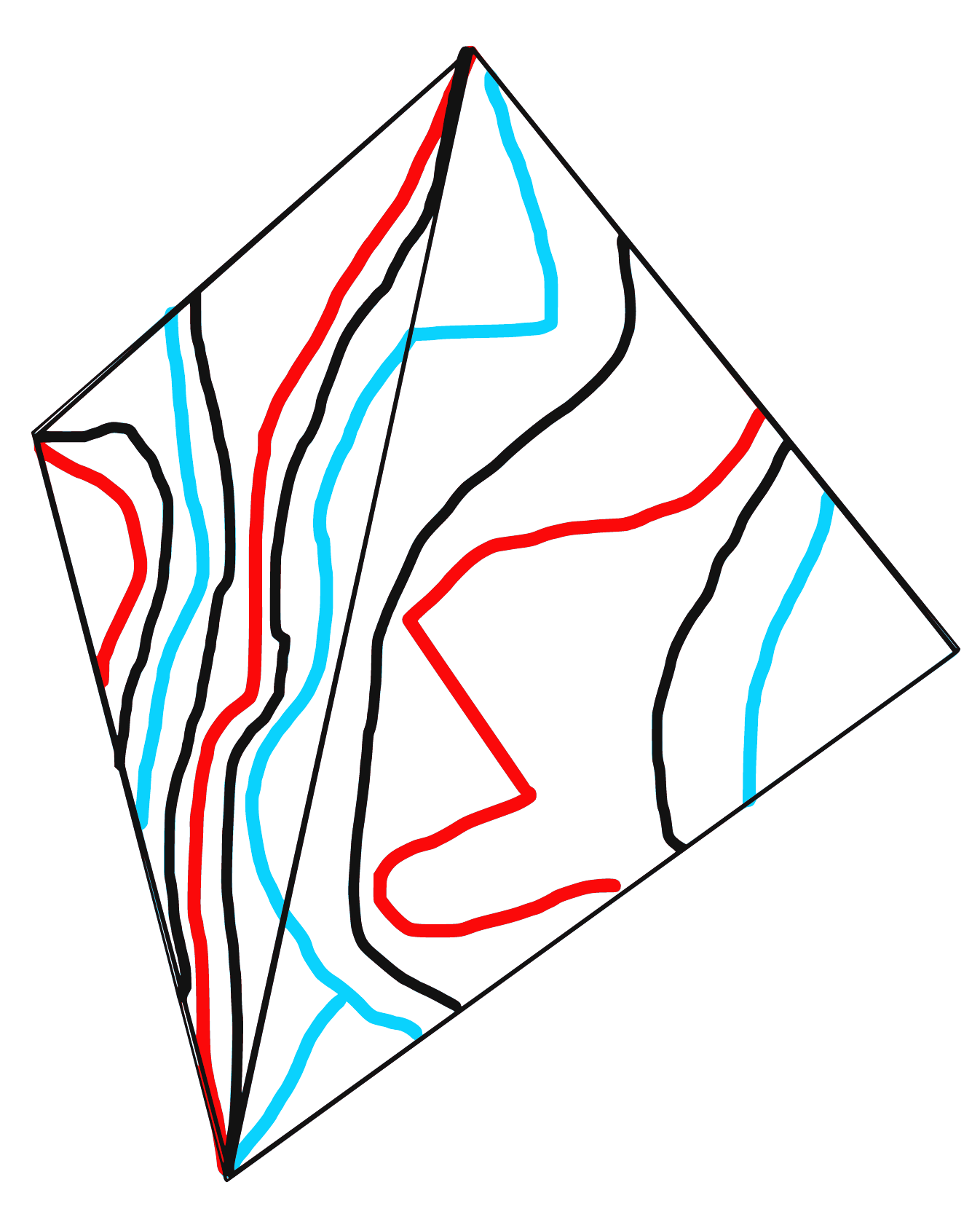}
      \includegraphics[width=\hsize]{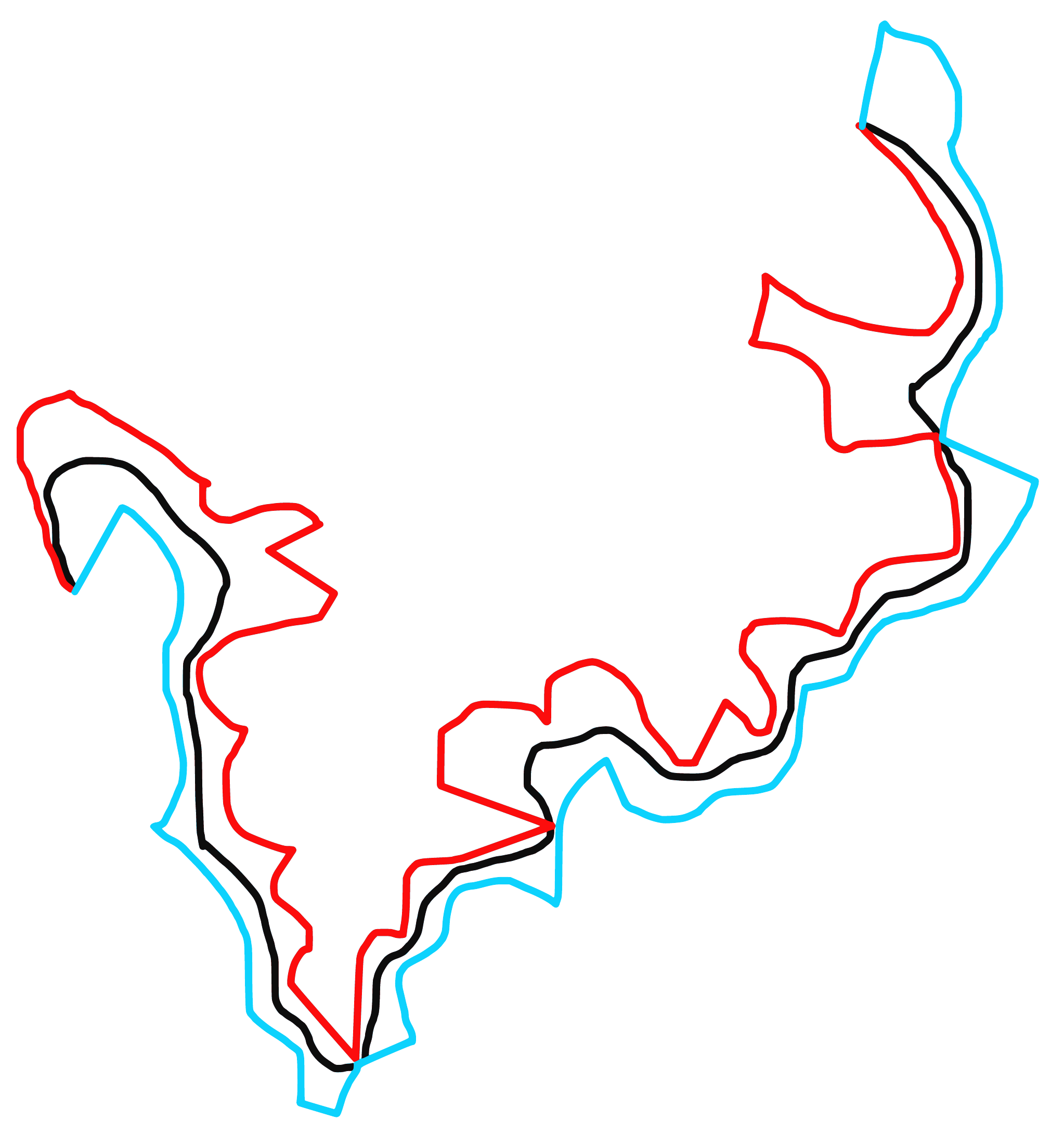}}
  \parbox[t]{0.29\textwidth}{\vspace{13mm}\includegraphics[width=\hsize,angle=-90]{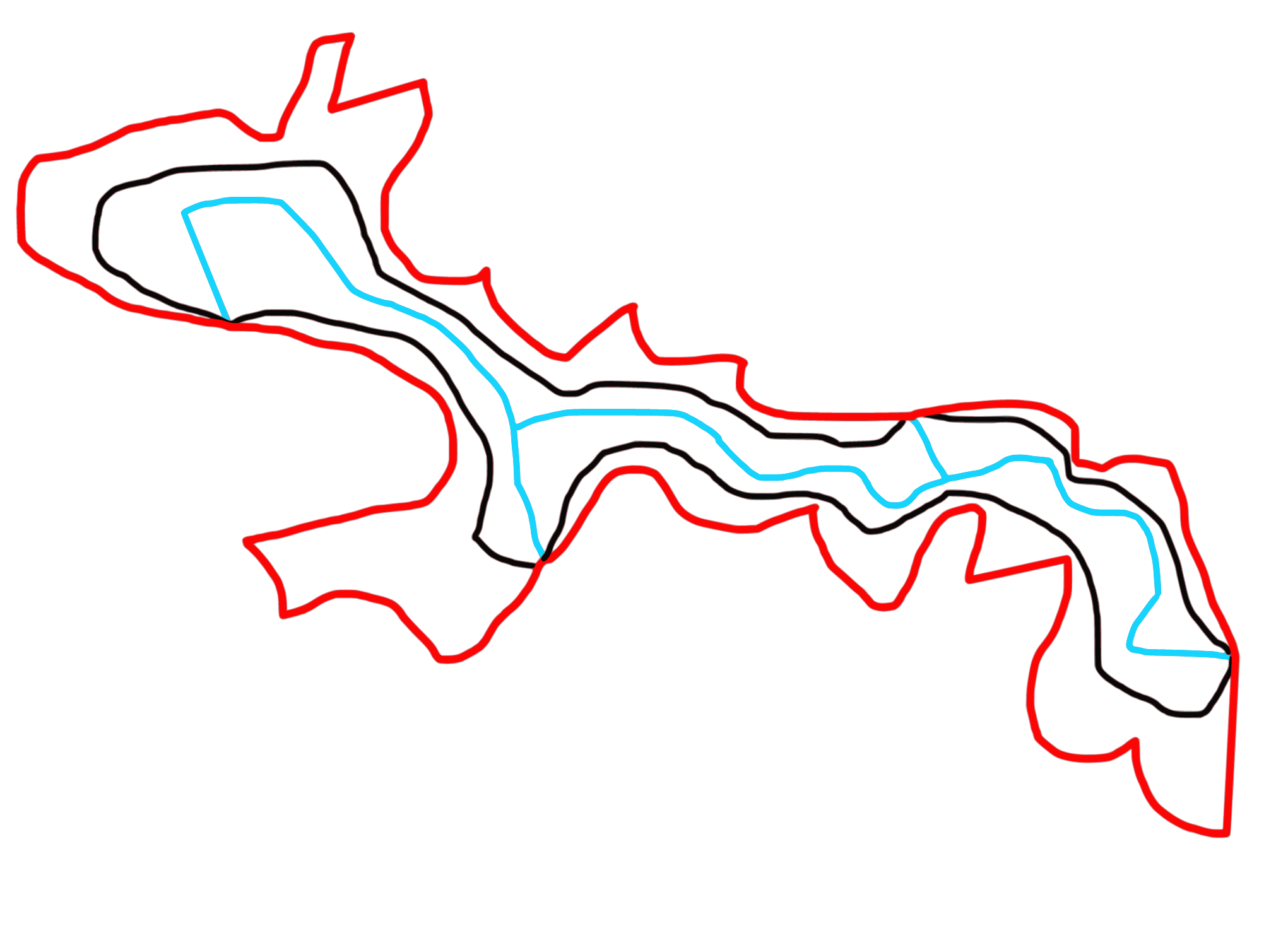}}\\
\begin{picture}(0,0)
\put(110,240){\mbox{$T$:}}
\put(10,220){\mbox{$N_1$: sea horse}}
\put(220,220){\mbox{$N_2$: weasel}}
\put(70,180){\mbox{\scalebox{1.8}{$\longleftarrow$}}}
\put(180,180){\mbox{\scalebox{1.8}{$\longrightarrow$}}}
\put(50,110){\mbox{\rotatebox{135}{\scalebox{1.8}{$\Longleftrightarrow$}}}}
\put(220,110){\mbox{\rotatebox{45}{\scalebox{1.8}{$\Longleftrightarrow$}}}}
\end{picture}
}
\caption{sea horse $\Leftrightarrow$ weasel}\label{fig:seahorse}
\end{figure}

\begin{figure}[h]
\noindent
\mbox{
\hspace{-0.8cm}
\parbox[t]{0.5\hsize}{\vspace{7mm}
\includegraphics[width=\hsize,angle=15]{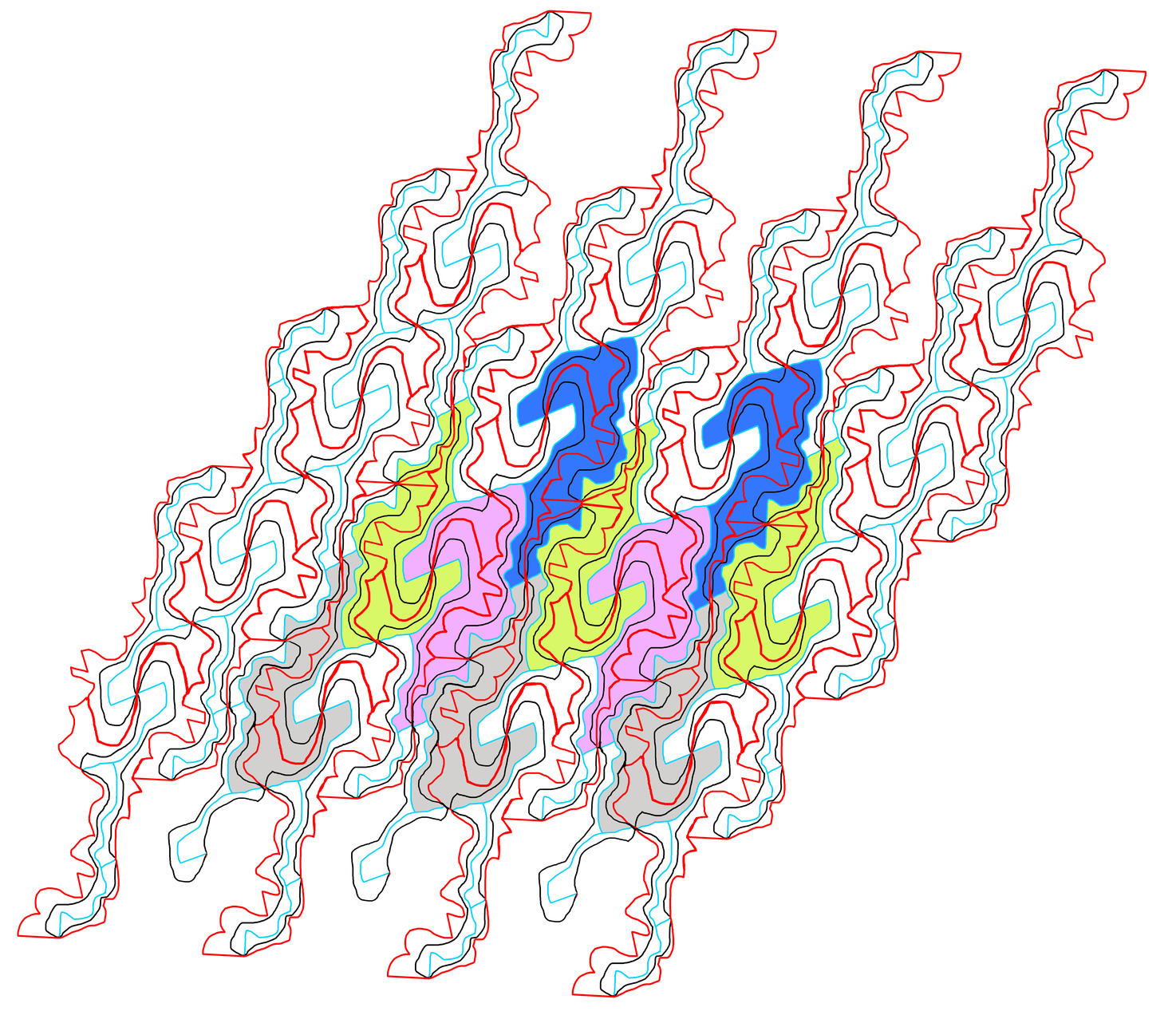}}
%\hspace{-2cm}
\parbox[t]{0.5\hsize}{\vspace{0mm}
%\clipbox{1cm 0cm 0cm 0cm}
{\includegraphics[width=\hsize,angle=15]{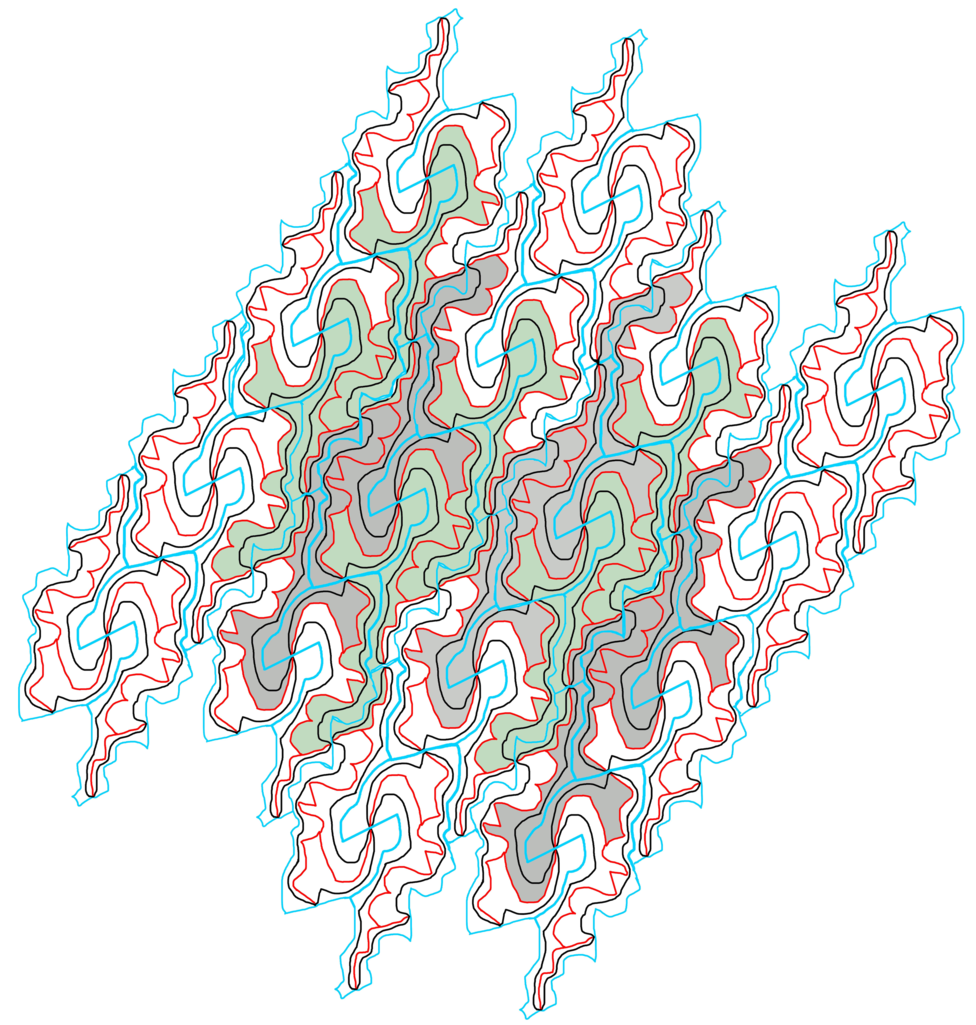}}}
}
\vspace{-1.5cm}
\caption{Tiling by sea horse and weasel }\label{fig:tiling}
\end{figure}

\begin{figure}[h]
\parbox[t]{\hsize}{
  \noindent
  \parbox[t]{0.39\textwidth}{\centering\vspace{15mm}\includegraphics[width=0.8\hsize]{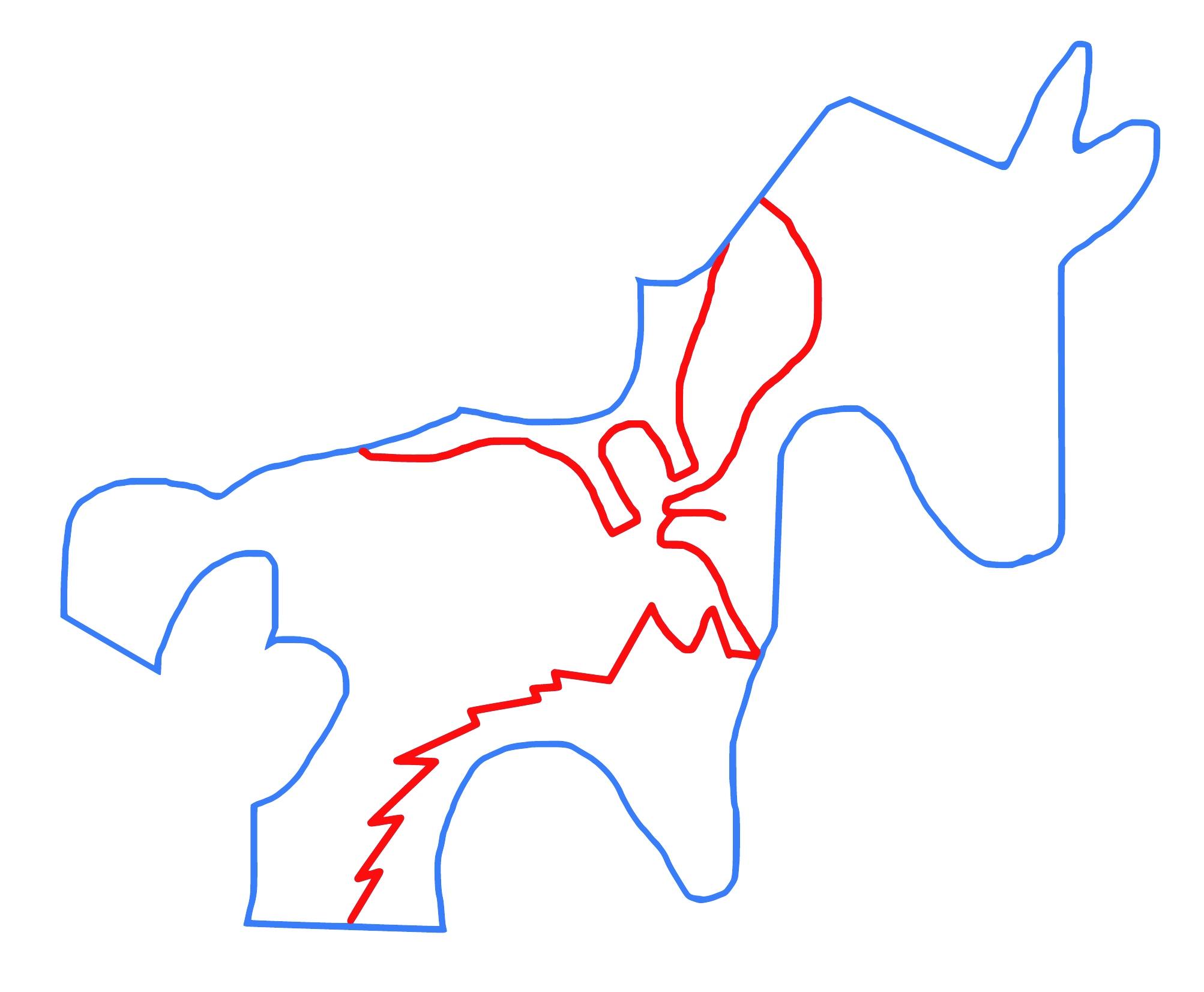}}
  \parbox[t]{0.20\textwidth}{\centering\vspace{0mm}\includegraphics[width=\hsize]{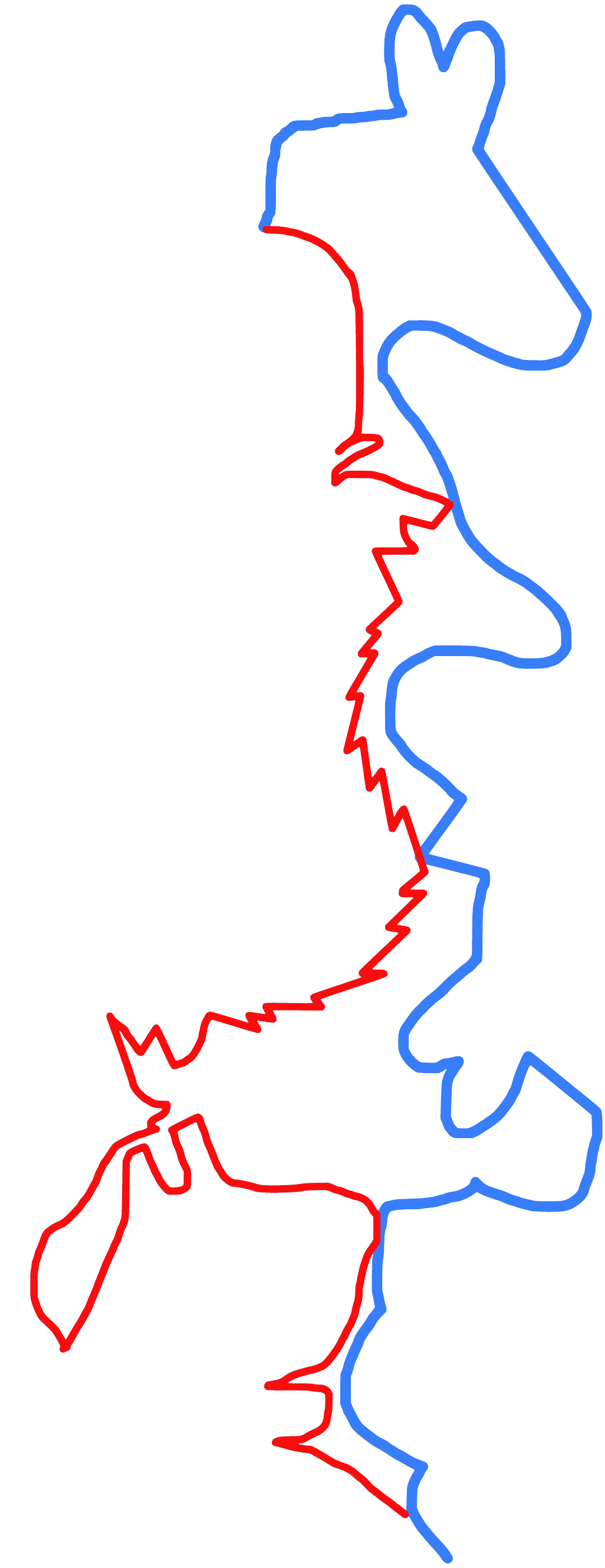}}
  \parbox[t]{0.39\textwidth}{\centering\vspace{20mm}\includegraphics[width=0.8\hsize]{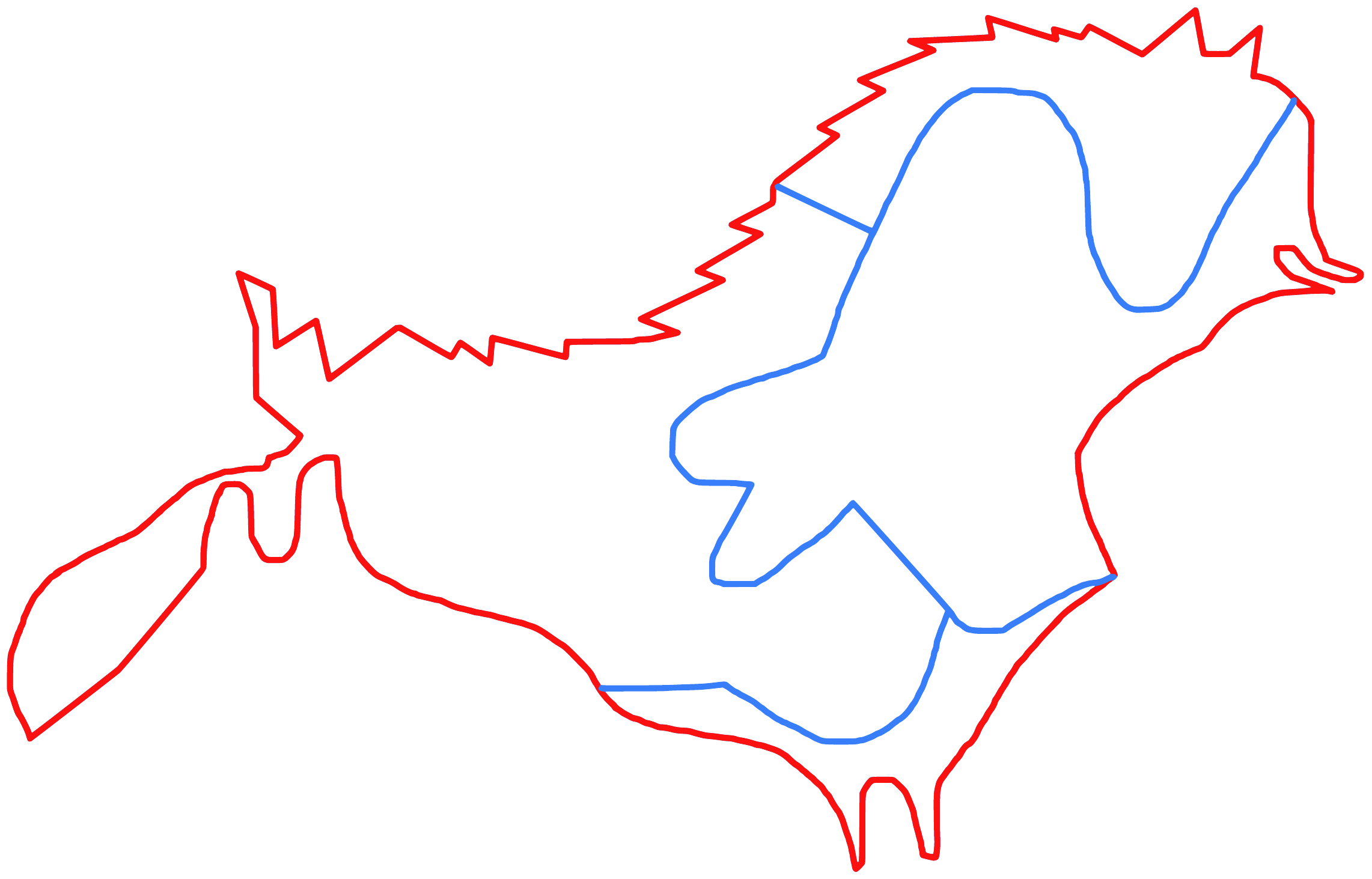}}\\
\begin{picture}(0,0)
\put(120,90){\mbox{\rotatebox{0}{\scalebox{1.8}{$\Longleftrightarrow$}}}}
\put(205,90){\mbox{\rotatebox{0}{\scalebox{1.8}{$\Longleftrightarrow$}}}}
\end{picture}
\vspace{-0,5cm}
  \parbox[t]{0.49\textwidth}{\vspace{0mm}\includegraphics[width=\hsize]{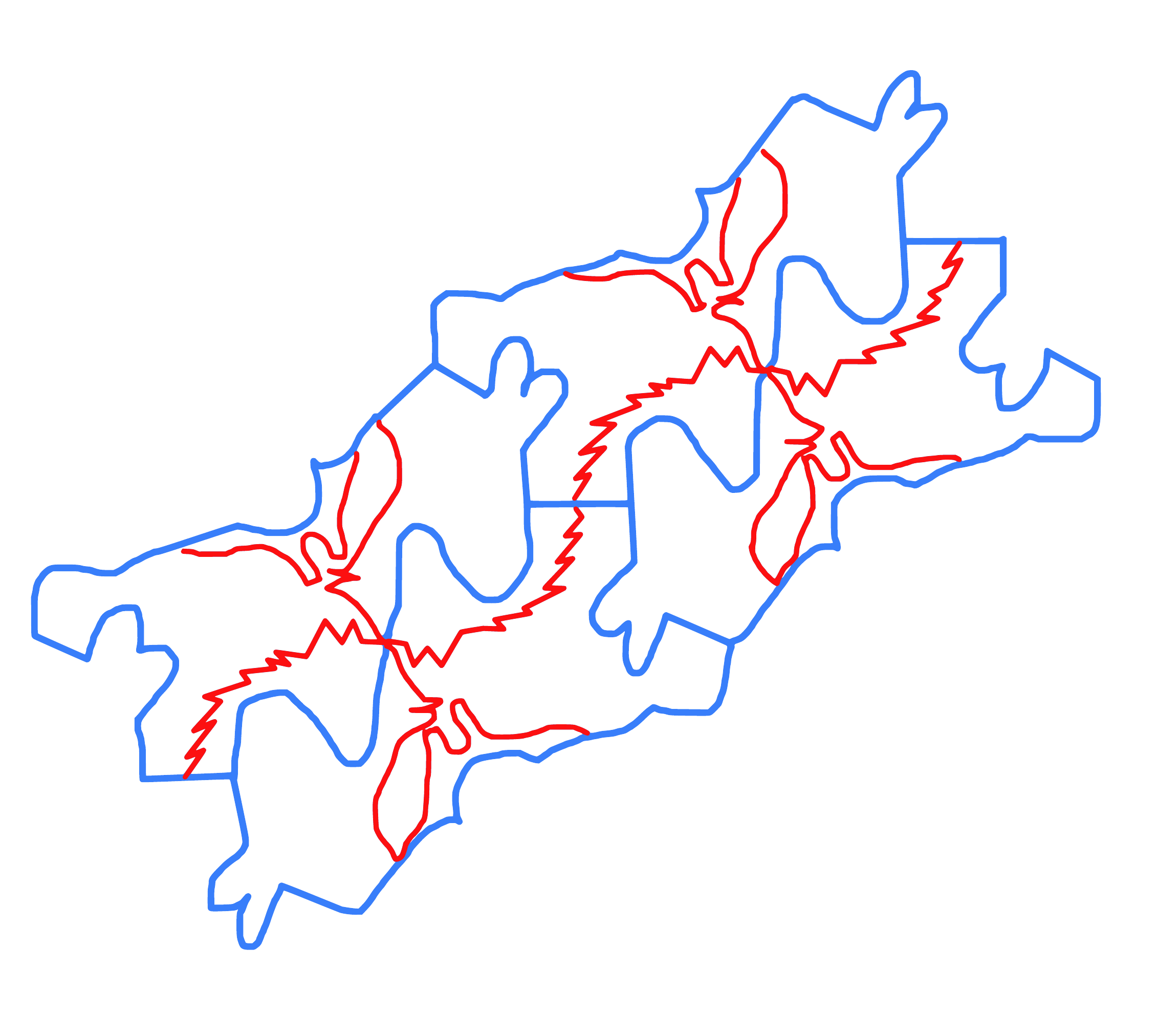}}
  \parbox[t]{0.49\textwidth}{\vspace{0mm}\includegraphics[width=\hsize]{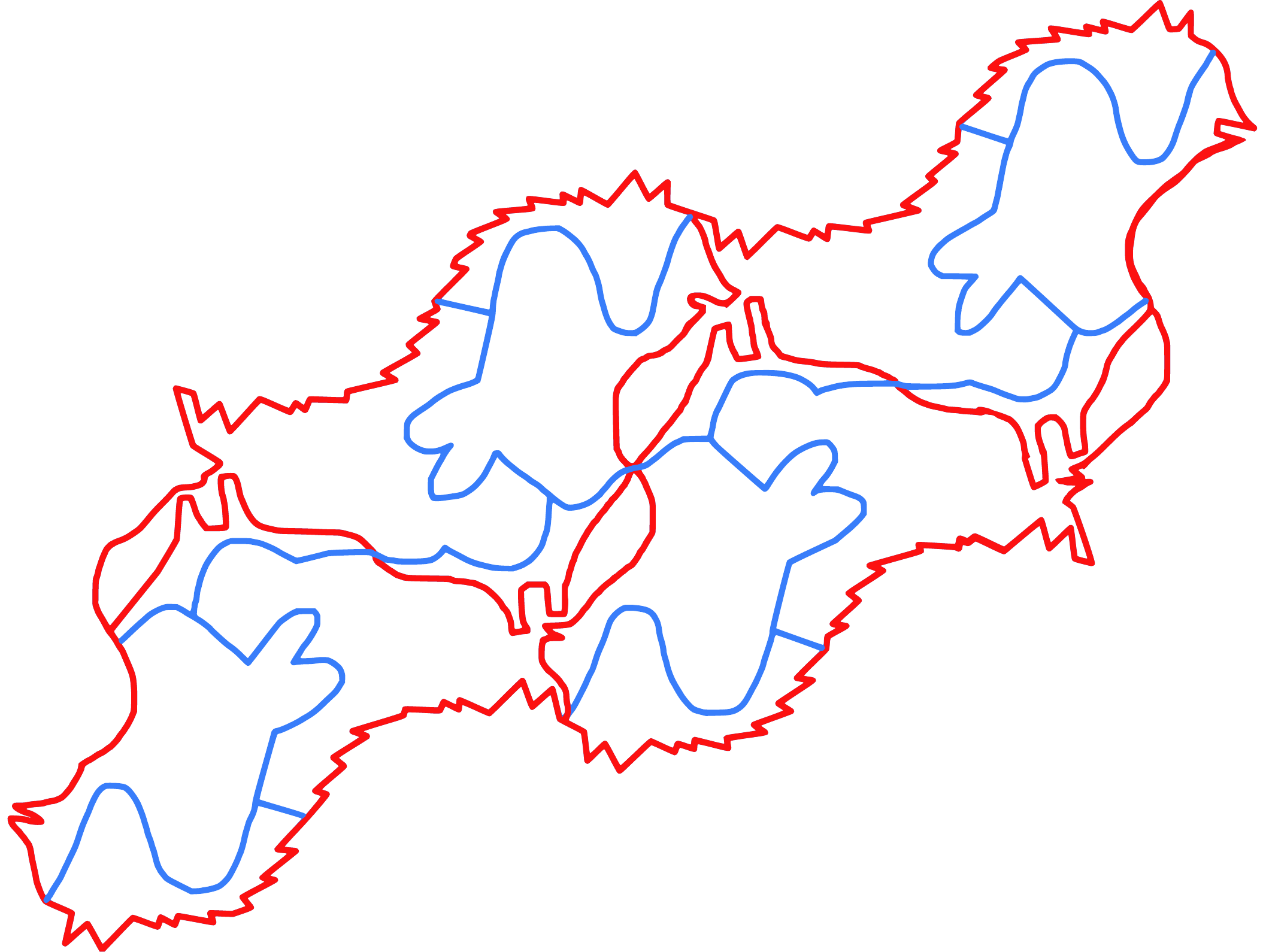}}\\
}
\caption{donkey $\Leftrightarrow$ fox}\label{fig:donkey}
\end{figure}

\begin{theorem}  
Let $D_1$ be an arbitrary dissection tree of an isotetrahedron
$T$. Then there exists a dissection tree $D_2$ of $T$ which doesn't
intersect $D_1$. % other than at vertices of $T$. 
The pair of nets $N_1$
and $N_2$ obtained by cutting along $D_1$ and $D_2$ is reversible, and
each $N_i$ ($i=1,2$) tiles the plane.
\end{theorem}

\begin{proof}  
By Theorem~\ref{thm:manynets}, there exists a $D_2$ for any $D_1$. Let
four vertices of $T$ be $v_k$ ($k=1,2,3,4$). 
Draw both $D_1$ and $D_2$ on two $T$s. Cut $T$ along $D_1$, and the
net $N_1$  inscribing $D_2$ is obtained. 
On the other hand, cut $T$ along $D_2$, and the net $N_2$ inscribing
$D_1$ is obtained (Fig.~\ref{fig:seahorse}).
As in Theorem~\ref{thm:reversible2}, dissect $N_1$ along $D_2$ (or
dissect $N_2$ along $D_1$) into four pieces $P_1$, $P_2$, $P_3$ and
$P_4$, and join then in sequence by three hinges on the perimeter of
$N_1$ like a chain. Fix one of the end pieces of the chain and rotate
the remaining pieces, then they form the net $N_2$ which is obtained
by cutting $T$ along $D_2$. Since each of $N_1$ and $N_2$ is a net of
an isotetrahedron, then both $N_1$ and $N_2$ are tessellative figures
(Fig.~\ref{fig:tiling}).  \qed
\end{proof}

\bibliographystyle{plain}
\bibliography{dissect}

\end{document}